\documentclass[journal]{IEEEtran}

\usepackage{cite}

\ifCLASSINFOpdf
   \usepackage[pdftex]{graphicx}
\else
   \usepackage[dvips]{graphicx}
\fi
\usepackage{mathrsfs}
\usepackage{amsmath}
\usepackage{amsthm}
\usepackage[dvipsnames]{xcolor}
\usepackage{xcolor}
\usepackage{bm}
\usepackage{mathtools}
\DeclarePairedDelimiter\ceil{\lceil}{\rceil}

\usepackage{amsfonts}
\newtheorem{remark}{Result}
\usepackage{color}
\usepackage{algorithm}
\usepackage{algorithmic}
\newtheorem{Theorem}{Theorem}
\newtheorem{Lemma}{Lemma}
\usepackage{stfloats}
\usepackage{array}
\usepackage{amsthm,amsmath,amssymb}
\usepackage{amssymb}
\usepackage{mathrsfs}
\usepackage{tikz}
\usetikzlibrary{positioning} 
\usepackage{xcolor} 
\usetikzlibrary {matrix}
\usetikzlibrary {shapes.geometric}
\usetikzlibrary{arrows.meta}
\usetikzlibrary{math}
\usetikzlibrary{patterns}

\title{Optimal Time of Arrival Estimation for MIMO Backscatter Channels}
\author{Chen~He,~\IEEEmembership{Member,~IEEE},~
Luyang Han ~and~
Z.~Jane~Wang,~\IEEEmembership{Fellow,~IEEE}
        }
        
\makeatletter
\def\thanks#1{\protected@xdef\@thanks{\@thanks
        \protect\footnotetext{#1}}}
\makeatother

\begin{document}
\thanks{Chen He (\textit{corresponding author}) is with School of Information Science and Technology, Northwest University, Xi' an, China (email: chenhe@nwu.edu.cn).}
\thanks{Luyang Han is with School of Information Science and Technology, Northwest University, Xi' an, China (email: hanluyang@stumail.nwu.edu.cn).}
\thanks{Z. Jane Wang is with Department of Electrical and Computer Engineering, University of British Columbia, Vancouver, BC, Canada (email: zjanew@ece.ubc.ca).}

\maketitle

\begin{abstract}
In this paper, we propose a novel time of arrival (TOA) estimator for multiple-input-multiple-output (MIMO) backscatter channels in closed form. The proposed estimator refines the estimation precision from the topological structure of the MIMO backscatter channels, and can considerably enhance the estimation accuracy. Particularly, we show that for the general $M \times N$ bistatic topology, the mean square error (MSE) is $\frac{M+N-1}{MN}\sigma^2_0$, and for the general $M \times M$ monostatic topology, it is $\frac{2M-1}{M^2}\sigma^2_0$ for the diagonal subchannels, and $\frac{M-1}{M^2}\sigma^2_0$ for the off-diagonal subchannels, where $\sigma^2_0$ is the MSE of the conventional least square estimator.
In addition, we derive the  Cramer-Rao lower bound (CRLB) for MIMO backscatter TOA estimation which indicates that the proposed estimator is optimal. Simulation results verify that the proposed TOA estimator can considerably improve both estimation and positioning accuracy, especially when the MIMO scale is large.
\end{abstract}


\section{Introduction}
Backscatter communications (BSC) is an emerging  technology in internet of things (IoT). The BSC system typically consists of a reader and multiple tags. Its hallmark is that the tag does not require internal battery, instead, it reflects the electromagnetic wave from the reader or the base station to transmit its information. This enables green, low-cost and sustainable communications and sensing technologies for future IoT. Since the tag is always passive, compared with the conventional communications, it is more difficult to achieve high-speed and reliable communication for BSC system and several works focus on improving the communication quality of BSC system, such as space-time coding \cite{he2020monostatic, luan2021better}, channel estimation\cite{mishra2019optimal},  signal detection\cite{zhang2018constellation}, etc.
Recently, localization and tracking are drawing increasing attention for backscatter communications. Existing localization scheme mainly based on  angle of arrival (AOA), received signal strength (RSS), time of arrival (TOA) and time difference of arrival (TDOA) \cite{sun2020eigenspace,liu2015rss,park2015closed, cheung2004least, gillette2008linear}.  The ranging accuracy depends on the accuracy of parameter estimation.
RSS estimation is mainly based on the model of pass loss which can be seriously  affected by the multipath. AOA estimation is mainly based on array antenna structure of receiver, which usually requires complex hardware. Therefore TOA or TDOA based localization and tracking is sometimes a good choice for BSC system \cite{sahinoglu2008ultra}. 



\begin{figure*}
\centerline{\includegraphics[width=1.98\columnwidth]{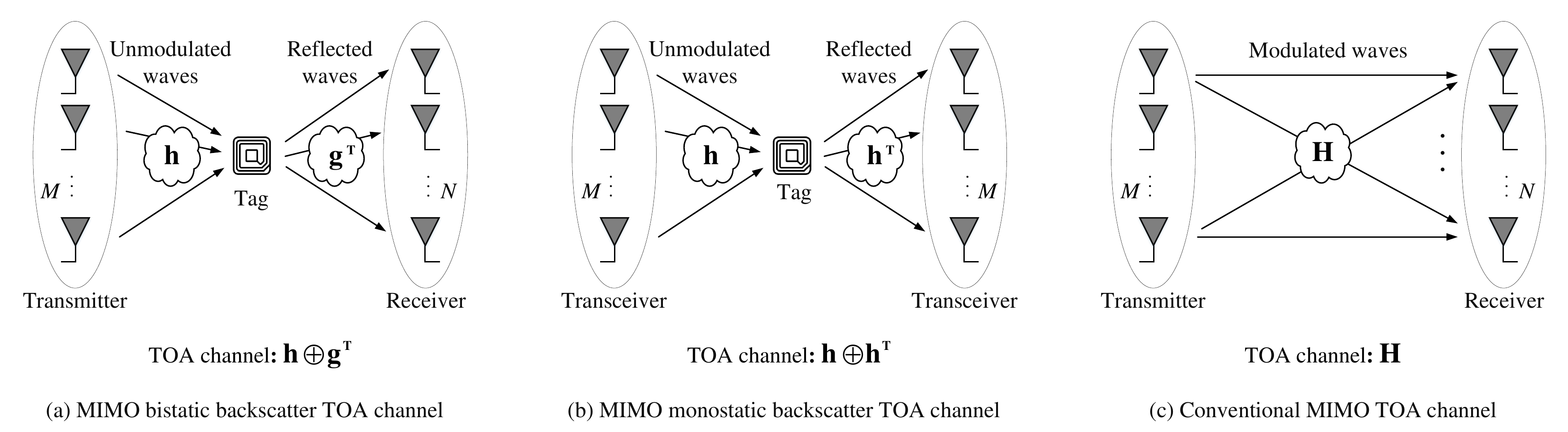}}
\caption{The illustration of the TOAs of the MIMO bistatic backscatter channel, the MIMO monostatic backscatter channel, and the conventional MIMO channel.}
\label{fig: BSCchannel}
\end{figure*}

\begin{figure*}

\begin{center}
\begin{tikzpicture}[scale=0.8, transmitter/.style={circle,draw=black!50,thick,inner sep=1.5mm},
tag/.style={rectangle,draw=black!50,thick,inner sep=2mm},
receiver/.style={circle,draw=black!50,fill=black!20,thick,inner sep=1.5mm},
transceiver/.style={circle,draw=black!50,pattern=vertical lines,thick,inner sep=1.5mm}]

\tikzstyle{every node}=[font=\small,scale=0.8]

\tikzmath{ \x = -2.5; };
\tikzmath{ \y = 2; };
\node[transmitter] at ( \x,\y) {}; 
\node[font=\fontsize{10}{10}\selectfont] at ( \x+1.2,\y) {Transmitter};
\node[tag] at ( \x+2.5,\y) {}; 
\node[font=\fontsize{10}{10}\selectfont] at ( \x+3.1,\y) {Tag};
\node[receiver] at ( \x+3.8,\y) {}; 
\node[font=\fontsize{10}{10}\selectfont] at ( \x+4.8,\y) {Receiver};

\node[transmitter] (t1) at ( \x-0.5,\y-1) {};
\node[transmitter] (t2) at ( \x+1.5,\y-1) {};
\node[tag] (t) at ( \x+0.5,\y-2)  {};
\node[receiver] (r1) at ( \x-0.5,\y-3)  {};
\node[receiver] (r2) at ( \x+1.5,\y-3) {};
\node[font=\fontsize{10}{10}\selectfont] at ( \x+3,\y-2) {has subchannels:};

\draw[thick,-{Stealth[width=4pt]}] (node cs:name=t1) -- (node cs:name=t)
node[left,pos=0.7,font=\fontsize{13}{13}\selectfont]{$h_1$};
\draw[thick,-{Stealth[width=4pt]}] (node cs:name=t2) -- (node cs:name=t) node[right,pos=0.7,font=\fontsize{13}{13}\selectfont]{$h_2$};
\draw[thick,-{Stealth[width=4pt]}] (node cs:name=t) -- (node cs:name=r1) node[right,pos=0.6,font=\fontsize{13}{13}\selectfont]{$g_1$};
\draw[thick,-{Stealth[width=4pt]}] (node cs:name=t) -- (node cs:name=r2) node[left,pos=0.6,font=\fontsize{13}{13}\selectfont]{$g_2$};

\node[transmitter] (t11) at ( \x+0.5,\y-4) {};
\node[tag] (tag1) at ( \x+1.5,\y-5)  {};
\node[receiver] (r11) at ( \x+0.5,\y-6)  {};

\draw[thick,-{Stealth[width=4pt]}] (node cs:name=t11) -- (node cs:name=tag1) node[left,pos=0.7,font=\fontsize{13}{13}\selectfont]{$h_1$};
\draw[thick,-{Stealth[width=4pt]}] (node cs:name=tag1) -- (node cs:name=r11) node[right,pos=0.6,font=\fontsize{13}{13}\selectfont]{$g_1$};
\node at ( \x+3,\y-2) {};
\node[font=\fontsize{13}{13}\selectfont] at ( \x+1.3,\y-6.7) {$t_1=h_1+g_1$};

\node[transmitter] (t12) at ( \x+5.5,\y-4) {};
\node[tag] (tag2) at ( \x+4.5,\y-5)  {};
\node[receiver] (r12) at ( \x+3.5,\y-6)  {};

\draw[thick,-{Stealth[width=4pt]}] (node cs:name=t12) -- (node cs:name=tag2) node[right,pos=0.7,font=\fontsize{13}{13}\selectfont]{$h_2$};
\draw[thick,-{Stealth[width=4pt]}] (node cs:name=tag2) -- (node cs:name=r12) node[right,pos=0.6,font=\fontsize{13}{13}\selectfont]{$g_1$};
\node[font=\fontsize{13}{13}\selectfont] at ( \x+4.5,\y-6.7) {$t_2=h_2+g_1$};

\node[transmitter] (t11) at ( \x+0.5,\y-7.5) {};
\node[tag] (tag1) at ( \x+1.5,\y-8.5)  {};
\node[receiver] (r11) at ( \x+2.5,\y-9.5)  {};

\draw[thick,-{Stealth[width=4pt]}] (node cs:name=t11) -- (node cs:name=tag1) node[left,pos=0.7,font=\fontsize{13}{13}\selectfont]{$h_1$};
\draw[thick,-{Stealth[width=4pt]}] (node cs:name=tag1) -- (node cs:name=r11) node[left,pos=0.6,font=\fontsize{13}{13}\selectfont]{$g_2$};
\node at ( \x+3,\y-2) {};
\node[font=\fontsize{13}{13}\selectfont] at ( \x+1.3,\y-6.7-3.5) {$t_3=h_1+g_2$};

\node[transmitter] (t12) at ( \x+5.5,\y-4-3.5) {};
\node[tag] (tag2) at ( \x+4.5,\y-5-3.5)  {};
\node[receiver] (r12) at ( \x+5.5,\y-6-3.5)  {};

\draw[thick,-{Stealth[width=4pt]}] (node cs:name=t12) -- (node cs:name=tag2) node[right,pos=0.7,font=\fontsize{13}{13}\selectfont]{$h_2$};
\draw[thick,-{Stealth[width=4pt]}] (node cs:name=tag2) -- (node cs:name=r12) node[left,pos=0.6,font=\fontsize{13}{13}\selectfont]{$g_2$};
\node[font=\fontsize{13}{13}\selectfont] at ( \x+4.5,\y-6.7-3.5) {$t_4=h_2+g_2$};

\node[font=\fontsize{10}{10}\selectfont] at ( \x+3,\y-11.3) {The subchannels are correlated in the pattern of:};
\node[font=\fontsize{13}{13}\selectfont] at ( \x+3,\y-11.8) {$t_1+t_4=t_2+t_3$};
\node[font=\fontsize{10}{10}\selectfont] at ( \x+3,\y-12.5) {(a) MIMO bistatic backscatter TOA channel};

\tikzmath{ \x = 5.5; };
\tikzmath{ \y = 2; };
\draw [pattern=north west lines, pattern color=black] (\x,\y) circle;
\node[transceiver] at ( \x+0.2,\y) {}; 
\node[font=\fontsize{10}{10}\selectfont] at ( \x+1.4,\y) {Transceiver};
\node[tag] at ( \x+3,\y) {}; 
\node[font=\fontsize{10}{10}\selectfont] at ( \x+3.6,\y) {Tag};

\node[transceiver] (t1) at ( \x-0.5,\y-1) {};
\node[transceiver] (t2) at ( \x+1.5,\y-1) {};
\node[tag] (t) at ( \x+0.5,\y-2)  {};
\node[transceiver] (r1) at ( \x-0.5,\y-3)  {};
\node[transceiver] (r2) at ( \x+1.5,\y-3) {};
\node[font=\fontsize{10}{10}\selectfont] at ( \x+3,\y-2) {has subchannels:};

\draw[thick,-{Stealth[width=4pt]}] (node cs:name=t1) -- (node cs:name=t) node[left,pos=0.7,font=\fontsize{13}{13}\selectfont]{$h_1$};
\draw[thick,-{Stealth[width=4pt]}] (node cs:name=t2) -- (node cs:name=t) node[right,pos=0.7,font=\fontsize{13}{13}\selectfont]{$h_2$};
\draw[thick,-{Stealth[width=4pt]}] (node cs:name=t) -- (node cs:name=r1) node[right,pos=0.7,font=\fontsize{13}{13}\selectfont]{$h_1$};
\draw[thick,-{Stealth[width=4pt]}] (node cs:name=t) -- (node cs:name=r2) node[left,pos=0.7,font=\fontsize{13}{13}\selectfont]{$h_2$};

\node[transceiver] (t11) at ( \x+0.5,\y-4) {};
\node[tag] (tag1) at ( \x+1.5,\y-5)  {};
\node[transceiver] (r11) at ( \x+0.5,\y-6)  {};

\draw[thick,-{Stealth[width=4pt]}] (node cs:name=t11) -- (node cs:name=tag1) node[left,pos=0.7,font=\fontsize{13}{13}\selectfont]{$h_1$};
\draw[thick,-{Stealth[width=4pt]}] (node cs:name=tag1) -- (node cs:name=r11) node[right,pos=0.7,font=\fontsize{13}{13}\selectfont]{$h_1$};
\node at ( \x+3,\y-2) {};
\node[font=\fontsize{13}{13}\selectfont] at ( \x+1.3,\y-6.7) {$t_1=h_1+h_1$};

\node[transceiver] (t12) at ( \x+5.5,\y-4) {};
\node[tag] (tag2) at ( \x+4.5,\y-5)  {};
\node[transceiver] (r12) at ( \x+3.5,\y-6)  {};

\draw[thick,-{Stealth[width=4pt]}] (node cs:name=t12) -- (node cs:name=tag2) node[right,pos=0.7,font=\fontsize{13}{13}\selectfont]{$h_2$};
\draw[thick,-{Stealth[width=4pt]}] (node cs:name=tag2) -- (node cs:name=r12) node[right,pos=0.7,font=\fontsize{13}{13}\selectfont]{$h_1$};
\node[font=\fontsize{13}{13}\selectfont] at ( \x+4.5,\y-6.7) {$t_2=h_2+h_1$};

\node[transceiver] (t11) at ( \x+0.5,\y-7.5) {};
\node[tag] (tag1) at ( \x+1.5,\y-8.5)  {};
\node[transceiver] (r11) at ( \x+2.5,\y-9.5)  {};

\draw[thick,-{Stealth[width=4pt]}] (node cs:name=t11) -- (node cs:name=tag1) node[left,pos=0.7,font=\fontsize{13}{13}\selectfont]{$h_1$};
\draw[thick,-{Stealth[width=4pt]}] (node cs:name=tag1) -- (node cs:name=r11) node[left,pos=0.7,font=\fontsize{13}{13}\selectfont]{$h_2$};
\node at ( \x+3,\y-2) {};
\node[font=\fontsize{13}{13}\selectfont] at ( \x+1.3,\y-6.7-3.5) {$t_3=h_1+h_2$};

\node[transceiver] (t12) at ( \x+5.5,\y-4-3.5) {};
\node[tag] (tag2) at ( \x+4.5,\y-5-3.5)  {};
\node[transceiver] (r12) at ( \x+5.5,\y-6-3.5)  {};

\draw[thick,-{Stealth[width=4pt]}] (node cs:name=t12) -- (node cs:name=tag2) node[right,pos=0.7,font=\fontsize{13}{13}\selectfont]{$h_2$};
\draw[thick,-{Stealth[width=4pt]}] (node cs:name=tag2) -- (node cs:name=r12) node[left,pos=0.7,font=\fontsize{13}{13}\selectfont]{$h_2$};
\node[font=\fontsize{13}{13}\selectfont] at ( \x+4.5,\y-6.7-3.5) {$t_4=h_2+h_2$};

\node[font=\fontsize{10}{10}\selectfont] at ( \x+3,\y-11) {The subchannels are correlated in the pattern of:};
\node[font=\fontsize{13}{13}\selectfont] at ( \x+3,\y-11.5) {$t_1+t_4=t_2+t_3$};
\node[font=\fontsize{13}{13}\selectfont] at ( \x+3,\y-12) {$t_2=t_3$};
\node[font=\fontsize{10}{10}\selectfont] at ( \x+3,\y-12.5) {(b) MIMO monostatic backscatter TOA channel};

\tikzmath{ \x = 13.5; };
\tikzmath{ \y = 2; };
\node[transmitter] at ( \x+0.3,\y) {};
\node[font=\fontsize{10}{10}\selectfont] at ( \x+1.5,\y) {Transmitter};
\node[receiver] at ( \x+3.3,\y) {}; 
\node[font=\fontsize{10}{10}\selectfont] at ( \x+4.3,\y) {Receiver};

\node[transmitter] (t1) at ( \x-0.5,\y-1) {};
\node[transmitter] (t2) at ( \x+1.5,\y-1) {};
\node[receiver] (r1) at ( \x-0.5,\y-3)  {};
\node[receiver] (r2) at ( \x+1.5,\y-3) {};
\node[font=\fontsize{10}{10}\selectfont] at ( \x+3.5,\y-2) {has subchannels:};

\draw[thick,-{Stealth[width=4pt]}] (node cs:name=t1) -- (node cs:name=r1) node[left,pos=0.5,font=\fontsize{13}{13}\selectfont]{$h_1$};
\draw[thick,-{Stealth[width=4pt]}] (node cs:name=t2) -- (node cs:name=r1) node[right,pos=0.4,font=\fontsize{13}{13}\selectfont]{$h_2$};
\draw[thick,-{Stealth[width=4pt]}] (node cs:name=t1) -- (node cs:name=r2) node[left,pos=0.4,font=\fontsize{13}{13}\selectfont]{$h_3$};
\draw[thick,-{Stealth[width=4pt]}] (node cs:name=t2) -- (node cs:name=r2) node[right,pos=0.5,font=\fontsize{13}{13}\selectfont]{$h_4$};

\node[transmitter] (t11) at ( \x+0.5,\y-4) {};
\node[receiver] (r11) at ( \x+0.5,\y-6)  {};

\draw[thick,-{Stealth[width=4pt]}] (node cs:name=t11) -- (node cs:name=r11) node[left,pos=0.5,font=\fontsize{13}{13}\selectfont]{$h_1$};
\node[font=\fontsize{13}{13}\selectfont] at ( \x+1,\y-6.7) {$t_1=h_1$};

\node[transmitter] (t12) at ( \x+5.5,\y-4) {};
\node[receiver] (r12) at ( \x+3.5,\y-6)  {};

\draw[thick,-{Stealth[width=4pt]}] (node cs:name=t12) -- (node cs:name=r12) node[right,pos=0.5,font=\fontsize{13}{13}\selectfont]{$h_2$};
\node[font=\fontsize{13}{13}\selectfont] at ( \x+4.5,\y-6.7) {$t_2=h_2$};

\node[transmitter] (t11) at ( \x+0.5,\y-7.5) {};
\node[receiver] (r11) at ( \x+2.5,\y-9.5)  {};

\draw[thick,-{Stealth[width=4pt]}] (node cs:name=t11) -- (node cs:name=r11) node[left,pos=0.5,font=\fontsize{13}{13}\selectfont]{$h_3$};

\node[font=\fontsize{13}{13}\selectfont] at ( \x+1,\y-6.7-3.5) {$t_3=h_3$};

\node[transmitter] (t12) at ( \x+5.5,\y-4-3.5) {};
\node[receiver] (r12) at ( \x+5.5,\y-6-3.5)  {};

\draw[thick,-{Stealth[width=4pt]}] (node cs:name=t12) -- (node cs:name=r12) node[right,pos=0.5,font=\fontsize{13}{13}\selectfont]{$h_4$};
\node[font=\fontsize{13}{13}\selectfont] at ( \x+4.5,\y-6.7-3.5) {$t_4=h_4$};

\node[font=\fontsize{10}{10}\selectfont] at ( \x+3,\y-11.5) {The subchannels are topologically independent.};
\node[font=\fontsize{10}{10}\selectfont] at ( \x+3,\y-12.5) {(c) Conventional MIMO TOA channel};

\end{tikzpicture}
\end{center}
\caption{Comparison of the topological structure of the MIMO backscatter TOA channels and the conventional MIMO TOA channel for $M=2$, and $N=2$.}
\label{fig: comBSCMIMO}
\end{figure*}
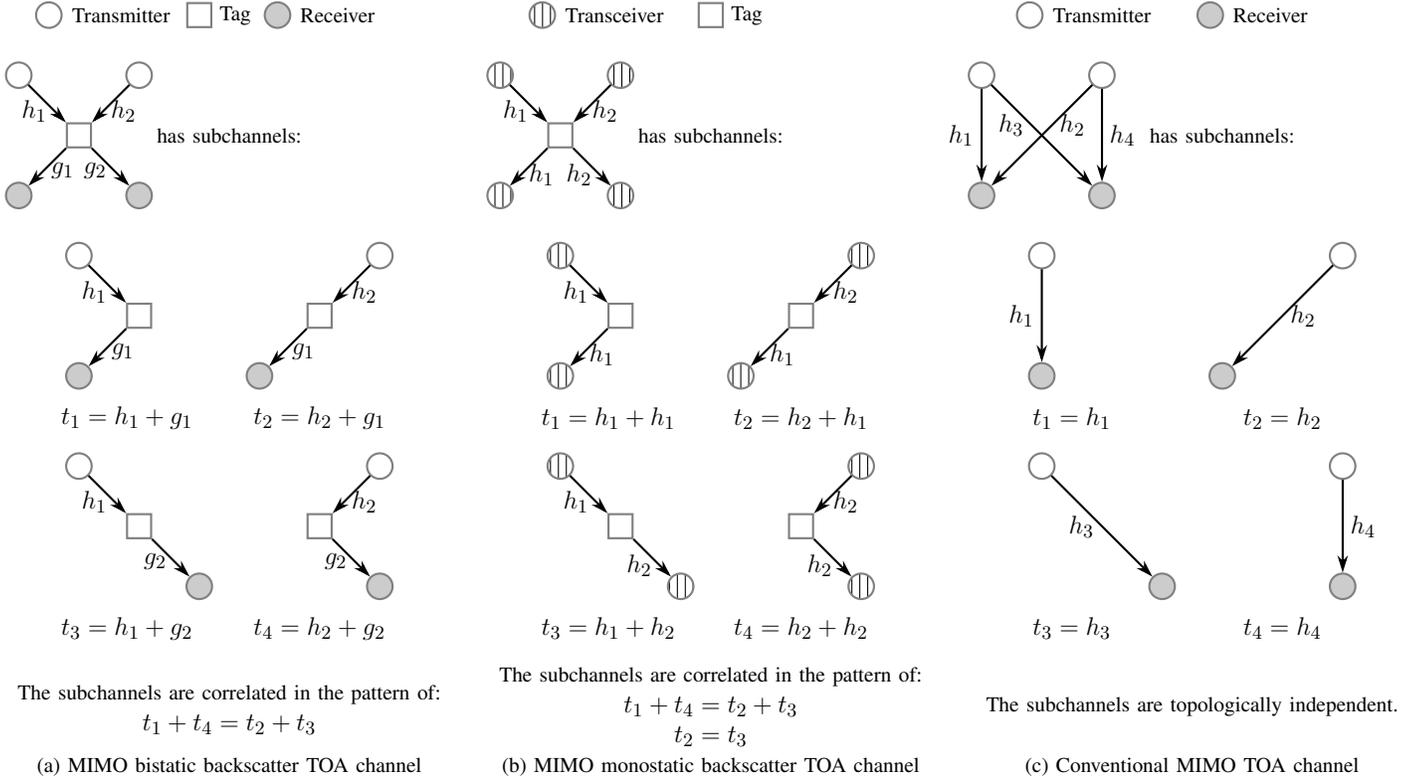

TOA estimation is usually obtained via least square (LS) in the conventional MIMO channel. In  backscatter channels, since the tag is a passive component which does not have any computational capacity,  the channel estimation can only be done at the receiver. As shown in Fig. \ref{fig: BSCchannel}, there are two major configurations for MIMO backscatter channels: the $M \times N$ bistatic configuration, which has $M$ transmitting antennas and $N$ receiving antennas, and the $M \times M$ monostatic configuration, which has $M$ antennas for both transmitting and receiving. For the bistatic configuration, the $M$ transmitting antennas transmit unmodulated waves to the tag, then the tag reflects waves to the receiving antennas by utilizing the energy of unmodulated waves. In this case, the transmitter and the receiver are different antennas and placed geographically apart. Therefore, the uplink channel (transmitter to tag) and the downlink channel (tag to receiver) are different, and the time delay of any TOA channel (transmitter to receiver) can be expressed as the sum of corresponding uplink delay and downlink delay. The monostatic configuration has a full-duplex architecture, i.e., the transmitter and receiver employ the same set of antennas,  so the uplink channel and the downlink channel can be modeled as the same channel.
The backscattering principal makes the MIMO structure fundamentally different from that of the conventional channels.


In this paper, we propose a novel TOA estimator by employing the topological structure of MIMO backscatter channels to refine the estimation accuracy. Therefore the proposed method in this paper is radically different from the works \cite{win2002characterization, d2008energy, bialer2011efficient, giorgetti2013time,kassas2021joint, shang2014ml, shamaei2021receiver,gifford2020impact, lee2002ranging, aditya2018survey,falsi2006time, xu2010toa}, which focused on the TOA estimation with certain receiving waveform, certain receiver structure, or with the consideration of the multipath effect. The proposed method is also radically different from the works in \cite{roemer2010tensor,sidiropoulos2017tensor,nion2010tensor,lioliou2008least,zhou2016channel,rong2012channel,mishra2019optimal,zhao2019channel}, which focused on the phase and amplitude estimation of MIMO channels, and cannot be employed for the TOA estimation of MIMO backscatter channels.


To the best of our knowledge, there is no work focus on the TOA estimation from the perspective of topological structure of MIMO backscatter channels
and the major contributions of this work are summarized as following:

\begin{itemize}
    \item  We propose a novel TOA estimator for MIMO backscatter channels in closed form that can significantly enhance the estimation precision compared with the conventional TOA estimator. Particularly, the proposed estimator employs the topological structure of the MIMO backscatter channels as constraints to refine the estimation accuracy.
    
    \item  We provide rigorous mathematical analysis to show  by how much the enhancement can be achieved by the proposed estimator. Particularly, for the general $M \times N$ bistatic topology, the mean square error (MSE) is $\frac{M+N-1}{MN}\sigma^2_0$, and for the general $M \times M$ monostatic topology, it is $\frac{2M-1}{M^2}\sigma^2_0$ for the diagonal subchannels, and $\frac{M-1}{M^2}\sigma^2_0$ for the off-diagonal subchannels, where $\sigma^2_0$ is the MSE of the conventional estimator.
    
    \item We derive the Cramer-Rao lower bounds (CRLB) for the TOA estimation of MIMO backscatter channels and show that the proposed estimator achieves the bound and hence it is optimal in terms of MSE.
\end{itemize}

We use boldfaced lower and upper cases for vectors and matrices, repectively. $\mathbf{X}^\intercal$, $\mathbf{X}^{-1}$, $\mathbf{X}^\dagger$, $\|\mathbf{X}\|$, and $[\mathbf{X}]_{i,j}$ denote the transpose, inverse, Moore-Penrose pseudoinverse, Frobenius norm, and the element at $i$-th row and $j$-th column of matrix $\mathbf{X}$, respectively.  $\operatorname{rank}(\mathbf{X})$ and $\operatorname{trace}(\mathbf{X})$ are the rank and trace of $\mathbf{X}$, respectively. $\mathbf{0}_{m\times n}$, $\mathbf{1}_{m\times n}$ and $\mathbf{I}_{n}$ represent the $m\times n$ zeros matrix, the $m\times n$ ones matrix and the $n\times n$ identity matrix, respectively. $\mathbb{R}$ and $\mathbb{C}$ denote real and complex number sets, respectively.  $\operatorname{E}\left\{\cdot\right\}$, $\operatorname{Var}\{\cdot\}$, and $\operatorname{Cov}\{\cdot, \cdot\}$ denote the expectation, the variance, and the covariance operations, respectively. $\ceil{\cdot}$ is the ceil operation. The operator $\text{vec}(\cdot)$ vectorizes a $M\times N$ matrix by column, while $\text{unvec}_{M\times N}(\cdot)$ is the inverse operation of $\text{vec}(\cdot)$. Finally, we use $\otimes$ to denote Kronecker product, and for $\mathbf{a}\in \mathbb{R}^{M\times 1}$, $\mathbf{b}\in \mathbb{R}^{1\times N}$, and $\mathbf{X}\in \mathbb{R}^{M\times N}$, the operator $\oplus$ is defined as $\mathbf{X}=\mathbf{a}\oplus \mathbf{b}=\mathbf{b}\oplus \mathbf{a}$ such that $[\mathbf{X}]_{i,j}=a_i+b_j$.

\section{MIMO backscatter TOA estimation}
\subsection{Channel model}
As shown in Fig. \ref{fig: BSCchannel}, the time delay for a general $M\times N$ bistatic backscatter channel ($M$ transmitting antennas, $N$ receiving antennas) is given by
\begin{align}
    \mathbf{T}= \Delta\cdot \bm{1}_{M \times N}+ \mathbf{h}\oplus\mathbf{g}^\intercal,
    \label{trueTOAbi}
\end{align}
where $\Delta$ is the time delay of the passive tag, $\mathbf{h}\in\mathbb{R}^{M\times 1}$ and $\mathbf{g}^\intercal\in\mathbb{R}^{1\times N}$ denote the  delays from the transmitting antennas to the passive tag and those from the passive tag to the receiving antennas, respectively. Similarly, for the general $M \times M$ monostatic topology, where $M$ antennas are employed for both transmitting and receiving, the time delay is given by 
\begin{align}
    \mathbf{T}= \Delta\cdot \bm{1}_{M \times M}+ \mathbf{h}\oplus\mathbf{h}^\intercal,
    \label{trueTOAmo}
\end{align}
where $\mathbf{h}\in\mathbb{R}^{M\times 1}$ represents the channel delays from the transceiving antennas to the passive tag. Fig. \ref{fig: comBSCMIMO} illustrates the topological structure of the $2 \times 2$ MIMO backscatter channels. As we can see that the subchannels in the MIMO bacskcatter channels are topologically correlated in certain patterns, while the subchannels in the conventional MIMO channel are topologically independent.

We define $\mathbf{Y} \in \mathbb{R}^{LM\times N}$ as the difference between the pilot receiving time and the pilot transmitting time, i.e., the observations of channel delay, where $L$ is the number of pilots of each transmitter ($L\geq 1$). The system model in TOA estimation can be written as
\begin{align}
    \mathbf{Y}=\mathbf{X}\mathbf{T}+\mathbf{W}, 
    \label{orginalOBJ}
\end{align}
where the TOA measurement errors $\mathbf{W}\in \mathbb{R}^{LM\times N}$ are assumed to be zero-mean Gaussian variables with variances $\sigma^2$ \cite{gholami2012improved, nguyen2016optimal,coluccia2017hybrid}, the TOA localization error is  assumed to be Gaussian \cite{park2015closed,cheung2004least}, and this is equivalent to TOA error being Gaussian \cite{gholami2012improved, nguyen2016optimal,coluccia2017hybrid}. $\mathbf{X}\in \mathbb{R}^{LM\times M}$ represents the pilots from the transmitting antennas, where $L$ is the length of the pilot. Note that the optimal $\mathbf{X}$ is given by $\mathbf{X}^\intercal\mathbf{X}=L\mathbf{I}_{M}$\cite{hassibi2003much}, so in this paper we set $\mathbf{X}=\mathbf{I}_{M}\otimes\mathbf{1}_{L\times 1}$.
The TOA model in (\ref{orginalOBJ}) can be rewrite in the vector form as
\begin{align}
    \mathbf{y}=\mathbf{S}\mathbf{t}+\mathbf{w},
\end{align}
 where $\mathbf{y}=\text{vec}(\mathbf{Y})$, $\mathbf{t}=\text{vec}(\mathbf{T})$ and $\mathbf{w}=\text{vec}(\mathbf{W})$,  $\mathbf{S}=\mathbf{I}_{N}\otimes\mathbf{X}=\mathbf{I}_{MN}\otimes\mathbf{1}_{L\times 1}$ is the transmitting matrix  corresponds to $\mathbf{X}$. 
It's also not hard to check that $\mathbf{S}$ satisfied $\mathbf{S}^\intercal\mathbf{S}=L\mathbf{I}_{MN}$.

\subsection{The proposed TOA estimator}
The conventional TOA estimation is given by \cite{kay1993fundamentals}
\begin{align}
    &\mathop{\arg\min}_\mathbf{t}  \|\mathbf{y}-\mathbf{S}\mathbf{t}\|^2,
\end{align}
which can be solved by least square (LS), i.e.,
\begin{align}
    \hat{ \mathbf{t}}=\mathbf{S}^\dagger\mathbf{y}=(\mathbf{S}^\intercal\mathbf{S})^{-1}\mathbf{S}^\intercal\mathbf{y}.\label{TOAconvcloseform}
\end{align}
The above estimator is optimal for the conventional MIMO channel when the noises on the receivers are independent and identically distributed, and can also be employed for the MIMO backscatter channels.
However, as we can see from Fig. \ref{fig: comBSCMIMO}, the subchannels are topologically correlated in certain patterns in both bistatic and monostatic channels, and the subchannels are topologically independent in the conventional the MIMO channel. It is not hard to see from Fig. \ref{fig: comBSCMIMO} that the topological structure is in linear form, and for the $M\times N$ bistatic topology, we have
\begin{align}
\mathbf{A}{\mathbf{t}}=\mathbf{0}_{(M-1)(N-1)\times1},
\end{align}
where $\mathbf{A}\in\{1,0,-1\}^{(M-1)(N-1)\times MN}$ is the correlation matrix.  For example, it is not hard to see that when $M=N=2$, $\mathbf{A}=\begin{pmatrix}
1 & -1 & -1 & 1
\end{pmatrix}$,
and when $M=3$, $N=2$, 
$\mathbf{A}=\begin{pmatrix}
1 & -1 & 0 & -1 & 1 & 0 \\
0 & 1 & -1 & 0 & -1 & 1 \\
\end{pmatrix}.$
In the next subsection, we will show how to generate the topological matrix $\mathbf{A}$ for general $M$ and $N$.

Inspired by this, we  employ such topological correlation as the constraints to the time delay vector $\mathbf{t}$ and  propose the following estimation for the bistatic configuration,
\begin{align}
\mathcal{OP}_{b}: \ \  &\mathop{\arg\min}_{\mathbf{t}} \ \ \| \mathbf{y}-\mathbf{S}\mathbf{t}\|^2 \nonumber\\
\operatorname{s.t.}\quad &
 \mathbf{A}{\mathbf{t}}=\mathbf{0}_{(M-1)(N-1)\times1}.
 \label{TDOAbiConsObjectF}
\end{align}
By employing Lagrange multiplier, the closed form of the proposed estimator for bistatic configuration is given by  \cite{kay1993fundamentals}

\begin{align}
\widetilde{\mathbf{t}}=&(\mathbf{I}_{MN}- (\mathbf{S}^\intercal\mathbf{S})^{-1}\mathbf{A}^\intercal(\mathbf{A}(\mathbf{S}^\intercal\mathbf{S})^{-1}\mathbf{A}^\intercal)^{-1}\mathbf{A})\mathbf{S}^\dagger\mathbf{y} \nonumber \\
=&(\mathbf{I}_{MN}- \mathbf{A}^\intercal(\mathbf{A}\mathbf{A}^\intercal)^{-1}\mathbf{A})\mathbf{S}^\dagger\mathbf{y} \nonumber \\
=&\mathbf{B}\mathbf{S}^\dagger\mathbf{y}=\mathbf{B}\hat{\mathbf{t}},
\label{TOAbicloseformVec}
\end{align}

where $\mathbf{B}=\mathbf{I}_{MN}- \mathbf{A}^\intercal(\mathbf{A}\mathbf{A}^\intercal)^{-1}\mathbf{A}$, and $\hat{\mathbf{t}}$ is the LS estimator in \eqref{TOAconvcloseform}.
The matrix form is given by
\begin{align}\label{TOAbicloseform}
    \widetilde{\mathbf{T}}=\text{unvec}_{M\times N}(\widetilde{\mathbf{t}}).
\end{align}

For the monostatic topology, since the uplink and the downlink are identical, i.e., $\mathbf{h}=\mathbf{g}$, we can treat it as a special case of the bistatic with $\mathbf{T}^\intercal=\mathbf{T}$, and the estimation can be formulated as
\begin{align}
\mathcal{OP}_{m}: \ \ &\mathop{\arg\min}_{\mathbf{t}} \ \ \| \mathbf{y}-\mathbf{S}\mathbf{t}\|^2 \nonumber\\
\operatorname{s.t.}\quad & \mathbf{A}{\mathbf{t}}=\mathbf{0}_{(M-1)(M-1)\times1}, \nonumber\\
\quad &\mathbf{T}^\intercal=\mathbf{T}.
\label{TDOAmoConsObjectF}
\end{align}
Under the first constraint solely, the solution has the same form as that for the bistatic:
\begin{align}
   \bar{\mathbf{t}}=\mathbf{B}\hat{ \mathbf{t}}=(\mathbf{I}_{MM}- \mathbf{A}^\intercal(\mathbf{A}\mathbf{A}^\intercal)^{-1}\mathbf{A}) \hat{ \mathbf{t}},
    \label{eq:mofirstcons}
\end{align}
With the consideration of the second constraint,  $\mathcal{OP}_{m}$ can be rewritten as
\begin{align}
    \mathcal{OP}_{m}: \ \ &\mathop{\arg\min}_{\mathbf{T}} \ \ \| \mathbf{T}-\bar{\mathbf{T}}\|^2 \nonumber\\
\operatorname{s.t.}
\quad &\mathbf{T}^\intercal=\mathbf{T},
\end{align}
where $\bar{\mathbf{T}}=\text{unvec}_{M\times M}(\bar{\mathbf{t}})$. Since $\mathbf{T}^\intercal=\mathbf{T}$, we have
\begin{align}\label{eq: OPM2}
    \| \mathbf{T}-\bar{\mathbf{T}}\|^2&=\operatorname{trace}(\mathbf{T}^\intercal\mathbf{T}-\bar{\mathbf{T}}^\intercal\mathbf{T}-\mathbf{T}^\intercal\bar{\mathbf{T}}+\bar{\mathbf{T}}^\intercal\bar{\mathbf{T}}) \nonumber \\
    &=\operatorname{trace}(\mathbf{T}\mathbf{T}-\bar{\mathbf{T}}^\intercal\mathbf{T}-\mathbf{T}\bar{\mathbf{T}}+\bar{\mathbf{T}}^\intercal\bar{\mathbf{T}}),
\end{align}
By taking the derivative of \eqref{eq: OPM2} and setting to zero, we have
\begin{align}
    2\mathbf{T}-(\bar{\mathbf{T}}+\bar{\mathbf{T}}^\intercal)=\mathbf{0}_{M\times M},
\end{align}
and the solution of $\mathcal{OP}_{m}$ is
\begin{align}\label{TOAmocloseform}
    \widetilde{\mathbf{T}}=\frac{1}{2}(\bar{\mathbf{T}}+\bar{\mathbf{T}}^\intercal).
\end{align}
and the corresponding vector form is $\widetilde{\mathbf{t}}=\text{vec}(\widetilde{\mathbf{T}})$.

\subsection{Generate correlation matrix $\mathbf{A}$ for general $M$ and $N$}
Here we provide an approach to generate $\mathbf{A}$ for general $M$ and $N$. To generate $\mathbf{A}$, we first define $(M-1)(N-1)$ submatrices in $\mathbf{T}$, in the way as shown in Fig. \ref{fig: illuA}. Each submatrix has a size of $2 \times 2$, and the $p$-th submatrices can be expressed as
\begin{align}
    \mathbf{D}_p=\begin{pmatrix}
    {[\mathbf{T}]}_{p_x,p_y} &  {[\mathbf{T}]}_{p_x,p_y+1} \\ {[\mathbf{T}]}_{p_x+1,p_y} &  {[\mathbf{T}]}_{p_x+1,p_y+1}
    \end{pmatrix},
\end{align}
where $p_x=p-(M-1)(p_y-1)$, $p_y=\ceil{\frac{p}{M-1}}$. 
Based on the topological structure of the bistatic channel, it is not hard to check that $\mathbf{D}_p$ satisfies
\begin{align}\label{eq: submatrix constraint}
\begin{pmatrix}1 &-1& -1& 1 \end{pmatrix}\text{vec}(\mathbf{D}_p)=0,
\end{align}
for all $p \in\{1,\cdots, (M-1)(N-1)\}$. This constraint is equivalent to
\begin{align}
\mathbf{a}_p\mathbf{t}=0,
\end{align}
where $\mathbf{a}_p$ is the $p$-th row of $\mathbf{A}$.
It is not hard to see that except the four non-zero elements $1$, $-1$, $-1$, $1$, all other elements of $\mathbf{a}_p$ are zeros.
So the key to generate $\mathbf{A}$ is to find the indices of $1$, $-1$, $-1$, $1$ for each $\mathbf{a}_p$. Based on the topological structure of the MIMO backscatter channels, it is not hard to verify that $[\mathbf{A}]_{p,q}=1$, $[\mathbf{A}]_{p,q+1}=-1$, $[\mathbf{A}]_{p,q+M}=-1$ and $[\mathbf{A}]_{p,q+M+1}=1$, where $q=p_x+M(p_y-1)=p+\lceil\frac{p}{M-1}\rceil-1$.
The corresponding routine to generate the constraint matrix $\mathbf{A}$ is summarized in Algorithm \ref{algo1}, where $\mathbf{A}\in\{1,0,-1\}^{(M-1)(N-1)\times MN}$.  Clearly, the rows of $\mathbf{A}$, which represent the submatrices in Fig. \ref{fig: illuA} are linearly independent, therefore the rank of $\mathbf{A}$ is $(M-1)(N-1)$. 
\begin{figure}
\centering

\begin{tikzpicture}[scale=0.9]
\tikzstyle{every node}=[font=\small,scale=0.9]

\node at (-4.6,0) {$\mathbf{T}=$};

\draw[red, fill=red!5, thick] (-3.7,1.9) rectangle (-1.1,0.7);
\draw[violet, fill=violet!5, thick] (-3.6,1.2) rectangle (-1,0.1);
\draw[blue, fill=blue!5, thick] (-3.8,-0.7) rectangle (-0.8,-1.8);
\draw[green, fill=green!5, thick] (-2.1,1.8) rectangle (0.3,0.6);
\draw[orange, fill=orange!5, thick] (0.3,-0.7) rectangle (3.8,-1.8);
\draw[red, thick] (-3.7,1.9) rectangle (-1.1,0.7);
\draw[violet, thick] (-3.6,1.2) rectangle (-1,0.1);

\draw[red, thick,->] (-2.5,1.9)--(-2.5,2.2);
\node [red] at (-2.5,2.4) {$\mathbf{D}_1$};
\draw[violet, thick,->] (-2.4,0.1)--(-2.4,-0.2);
\node [violet] at (-2.4,-0.4) {$\mathbf{D}_2$};
\draw[blue, thick,->] (-2.5,-1.8)--(-2.5,-2.1);
\node [blue] at (-2.4,-2.3) {$\mathbf{D}_{M-1}$};
\draw[green, thick,->] (-0.3,1.8)--(-0.3,2.1);
\node [green] at (-0.3,2.3) {$\mathbf{D}_{M}$};
\draw[orange, thick,->] (2,-1.8)--(2,-2.1);
\node [orange] at (2,-2.3) {$\mathbf{D}_{(M-1)(N-1)}$};

\matrix [matrix of math nodes,left delimiter=(,right delimiter=)]
{
{[\mathbf{T}]_{1,1}} & {[\mathbf{T}]_{1,2}} & {[\mathbf{T}]_{1,3}} & \dots & {[\mathbf{T}]_{1,N}} \\
{[\mathbf{T}]_{2,1}} & {[\mathbf{T}]_{2,2}} & {[\mathbf{T}]_{2,3}} & \dots & {[\mathbf{T}]_{2,N}} \\
{[\mathbf{T}]_{3,1}} & {[\mathbf{T}]_{3,2}} & {[\mathbf{T}]_{3,3}} & \dots & {[\mathbf{T}]_{3,N}} \\
\vdots &  & \ddots &  & \vdots \\
{[\mathbf{T}]_{M-1,1}} & {[\mathbf{T}]_{M-1,2}} &  \dots & {[\mathbf{T}]_{M-1,N-1}} & {[\mathbf{T}]_{M-1,N}} \\
{[\mathbf{T}]_{M,1}} & {[\mathbf{T}]_{M,2}} &  \dots & {[\mathbf{T}]_{M,N-1}} & {[\mathbf{T}]_{M,N}} \\
};
\end{tikzpicture}
\caption{Submatrices in $\mathbf{T}$.}
\label{fig: illuA}
\end{figure}
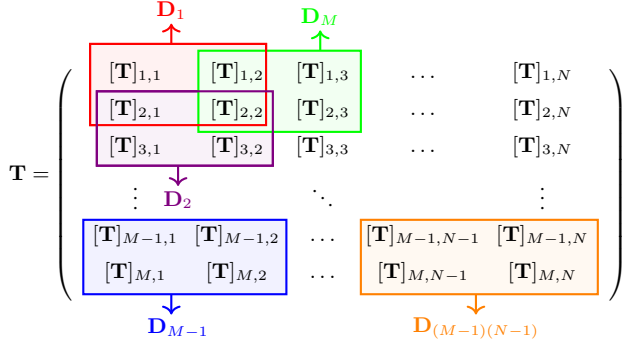

Now we show that the correlation matrix $\mathbf{A}$ generated by Algorithm \ref{algo1} contains all the topological information of MIMO backscatter TOA channels, i.e.,  $\mathbf{A}\mathbf{t}=\mathbf{0}_{(M-1)(N-1)\times1}$ is equivalent to $\mathbf{T}=\Delta+ \mathbf{h}\oplus\mathbf{g}^\intercal$.
\begin{Lemma}
For the correlation matrix $\mathbf{A}$ generated by Algorithm \ref{algo1}, $\mathbf{A}\mathbf{t}=\mathbf{0}_{(M-1)(N-1)\times1}$ and $ \mathbf{T}=\Delta+ \mathbf{h}\oplus\mathbf{g}^\intercal$ are equivalent.
\label{lemma: equivalence}
\end{Lemma}
\begin{proof}
It is easy to see that $\mathbf{T}=\Delta+ \mathbf{h}\oplus\mathbf{g}^\intercal \Rightarrow \mathbf{A}\mathbf{t}=\mathbf{0}_{(M-1)(N-1)\times1}$. So we only need to prove $\mathbf{T}=\Delta+ \mathbf{h}\oplus\mathbf{g}^\intercal \Leftarrow \mathbf{A}\mathbf{t}=\mathbf{0}_{(M-1)(N-1)\times1}$. From \eqref{eq: submatrix constraint} it is not hard to see that
\begin{align}\label{lemma: equivalence equation}
    [\mathbf{T}]_{1,1}+[\mathbf{T}]_{i,j} =[\mathbf{T}]_{i,1}+[\mathbf{T}]_{1,j},
\end{align}
where $i\in\{1,2,\dots,M\}$ and $j\in\{1,2,\dots,N\}$. According to \eqref{lemma: equivalence equation}, the delay matrix $\mathbf{T}$ can be written as 
\begin{align}
    \mathbf{T}=\underline{\mathbf{t}}_1\oplus(\mathbf{t}_1-[\mathbf{T}]_{1,1}\cdot \bm{1}_{M \times 1}),
\end{align}
where $\mathbf{t}_1=([\mathbf{T}]_{1,1}, [\mathbf{T}]_{1,2}, \dots, [\mathbf{T}]_{1,N})$ is the first row of $\mathbf{T}$ and $\underline{\mathbf{t}}_1=([\mathbf{T}]_{1,1}, [\mathbf{T}]_{2,1}, \dots, [\mathbf{T}]_{M,1})^\intercal$ is the first column of $\mathbf{T}$. Then, we have
\begin{align}
    \mathbf{T}=&\Delta\cdot \bm{1}_{M \times N}+(\underline{\mathbf{t}}_1-(\Delta+g_1)\cdot \bm{1}_{M \times 1})
    \nonumber \\
    &\oplus(\mathbf{t}_1+(g_1-[\mathbf{T}]_{1,1})\cdot \bm{1}_{M \times 1}) \nonumber \\
    =&\Delta\cdot \bm{1}_{M \times N}+(\underline{\mathbf{t}}_1-(\Delta+g_1)\cdot \bm{1}_{M \times 1})
    \nonumber \\
    &\oplus(\mathbf{t}_1-(\Delta+h_1)\cdot \bm{1}_{M \times 1}) \nonumber \\
    =&\Delta\cdot \bm{1}_{M \times N}+ \mathbf{h}\oplus\mathbf{g}^\intercal.
\end{align}
Therefore, $\mathbf{T}=\Delta+ \mathbf{h}\oplus\mathbf{g}^\intercal \Leftarrow \mathbf{A}\mathbf{t}=\mathbf{0}_{(M-1)(N-1)\times1}$ also holds.
\end{proof}


\begin{algorithm}
\caption{- Generating Correlation Matrix $\mathbf{A}$}
\begin{algorithmic}[1]
\label{algo1} \REQUIRE  $M$, $N$ (for monostatic $M=N$).
\ENSURE The constraint matrix $\mathbf{A}$.
\STATE Initialize $\mathbf{A}=\mathbf{0}_{(M-1)(N-1)\times MN}.$
\FOR{$p=1$ to $(M-1)(N-1)$} 
\STATE $q=p+\lceil\frac{p}{M-1}\rceil-1$;
\STATE ${[\mathbf{A}]}_{p, q}=1$;
\STATE ${[\mathbf{A}]}_{p, q+1}=-1$;
\STATE ${[\mathbf{A}]}_{p, q+M}=-1$;
\STATE ${[\mathbf{A}]}_{p, q+M+1}=1$;
\ENDFOR
\end{algorithmic}
\end{algorithm}

\subsection{The weighting matrix $\mathbf{B}$}
With $\mathbf{A}$ being generated, now we study the matrix $\mathbf{B}=\mathbf{I}_{MN}-\mathbf{A}^\intercal(\mathbf{A}\mathbf{A}^\intercal)^{-1}\mathbf{A}$ for general $M$ and $N$. It is clear that the finally estimator is the linear transformation of the original LS estimator and $\mathbf{B}$ can be treated as a weighting matrix, then we provide and prove the following results for $\mathbf{B}$ via Lemma \ref{lemma_1} and Lemma \ref{Lemma_z}. 
\begin{Lemma}
For $\mathbf{B}=\mathbf{I}_{MN}-\mathbf{A}^\intercal(\mathbf{A}\mathbf{A}^\intercal)^{-1}\mathbf{A}$, we have

\noindent
\textit{(1)} $\mathbf{B}\mathbf{1}_{MN\times1}=\mathbf{1}_{MN\times1}$.

\noindent
\textit{(2)} $\mathbf{B}\mathbf{A}^\intercal=\mathbf{0}_{MN \times (M-1)(N-1)}$.

\noindent
\textit{(3)} $\mathbf{B}$ is idempotent, i.e., $\mathbf{B}^2=\mathbf{B}$ 

\noindent
\textit{(4)} $\mathbf{B}$ is symmetric, i.e., $\mathbf{B}^\intercal=\mathbf{B}$.
\label{lemma_1}
\end{Lemma}

\begin{proof} Since $\mathbf{B}=\mathbf{I}_{MN}-\mathbf{A}^\intercal(\mathbf{A}\mathbf{A}^\intercal)^{-1}\mathbf{A}$, we have
\begin{align}
    \mathbf{B}\mathbf{1}_{MN\times1}=&\mathbf{I}_{MN}\mathbf{1}_{MN\times1}-\mathbf{A}^\intercal(\mathbf{A}\mathbf{A}^\intercal)^{-1}\mathbf{A}\mathbf{1}_{MN\times1} \nonumber \\
    =&\mathbf{I}_{MN}\mathbf{1}_{MN\times1}=\mathbf{1}_{MN\times1}, \nonumber \\
    \mathbf{B}\mathbf{A}^\intercal=&\mathbf{I}_{MN}\mathbf{A}^\intercal-\mathbf{A}^\intercal(\mathbf{A}\mathbf{A}^\intercal)^{-1}\mathbf{A}\mathbf{A}^\intercal \nonumber \\
   =&\mathbf{A}^\intercal-\mathbf{A}^\intercal=\mathbf{0}_{MN \times (M-1)(N-1)}, \nonumber \\
   \mathbf{B}^2=&\mathbf{I}_{MN}-2\mathbf{A}^\intercal(\mathbf{A}\mathbf{A}^\intercal)^{-1}\mathbf{A}\nonumber \\
    &+\mathbf{A}^\intercal(\mathbf{A}\mathbf{A}^\intercal)^{-1}\mathbf{A}\mathbf{A}^\intercal(\mathbf{A}\mathbf{A}^\intercal)^{-1}\mathbf{A} \nonumber \\
    =& \mathbf{I}_{MN}-\mathbf{A}^\intercal(\mathbf{A}\mathbf{A}^\intercal)^{-1}\mathbf{A}=\mathbf{B}, \nonumber \\
    \mathbf{B}^\intercal =&(\mathbf{I}_{MN}-\mathbf{A}^\intercal(\mathbf{A}\mathbf{A}^\intercal)^{-1}\mathbf{A})^\intercal=\mathbf{B}. 
\end{align}
Therefore Lemma \ref{lemma_1} holds.
\end{proof}
Based on the symmetric property of the channel, for any given subchannel $h_i+g_j$, we can classify the subchannels of the entire MIMO backscatter channel into four types:

\noindent
$\bullet$ Type 1: the subchannel containing both $h_i$ and $g_j$, i.e., the subchannel $h_i+g_j$ itself;

\noindent
$\bullet$ Type 2: the subchannels containing $h_i$ but not  $g_j$;

\noindent
$\bullet$ Type 3: the subchannels containing  $g_j$ but not  $h_i$;

\noindent
$\bullet$ Type 4: the subchannels containing neither $h_i$ nor $g_j$.

Note that $\mathbf{B}$ can be treated as the weighting matrix of the original LS estimator, therefore the elements of $\mathbf{B}$ can also be classified into four types as the above. For any subchannel $h_i+g_j$, we define the weights for $h_i+g_j$ from the Type 1, 2, 3 and 4 subchannels as $\alpha_1$, $\alpha_2$, $\alpha_3$, and $\alpha_4$, respectively.
For example, when $M=N=2$, $\mathbf{B}$ is in the following form 
\begin{align}
\mathbf{B}=\begin{pmatrix}    \alpha_1 & \alpha_3 & \alpha_2 & \alpha_4 \\  \alpha_3 & \alpha_1  & \alpha_4 & \alpha_2 \\ \alpha_2 &  \alpha_4  & \alpha_1  & \alpha_3 \\      \alpha_4   & \alpha_2 &  \alpha_3 &\alpha_1      \end{pmatrix},
\end{align}
and when $M=2$, $N=3$,
\begin{align}
\mathbf{B}=\begin{pmatrix}    \alpha_1& \alpha_3 & \alpha_2 & \alpha_4 & \alpha_2 & \alpha_4 \\  \alpha_3 & \alpha_1 & \alpha_4 & \alpha_2 & \alpha_4 & \alpha_2 \\ \alpha_2 &  \alpha_4  & \alpha_1 & \alpha_3 &\alpha_2 & \alpha_4\\      \alpha_4   & \alpha_2 &  \alpha_3 &\alpha_1&\alpha_4   & \alpha_2
\\      \alpha_2   & \alpha_4 &  \alpha_2 &\alpha_4 &\alpha_1& \alpha_3
\\      \alpha_4   & \alpha_2 &  \alpha_4 &\alpha_2 &\alpha_3 & \alpha_1\end{pmatrix}.
\end{align}
Now we prove the following Lemma to characterize the values of that $\alpha_1$, $\alpha_2$, $\alpha_3$, and $\alpha_4$.

\begin{Lemma}
For the bistatic configuration, we have $\alpha_1=\frac{M+N-1}{MN}$,  $\alpha_2=\frac{M-1}{MN}$, $\alpha_3=\frac{N-1}{MN}$, $\alpha_4=-\frac{1}{MN}$, for any arbitrary $M$ and $N$.  For the monostatic configuration ($M=N$), we have $\alpha_1=\frac{2M-1}{M^2}$, $\alpha_2=\alpha_3=\frac{M-1}{M^2}$ and $\alpha_4=-\frac{1}{M^2}$ for any arbitrary $M$. 
\label{Lemma_z}
\end{Lemma}

\begin{proof}
First, we consider the bistatic configuration and show that $\alpha_1=\frac{M+N-1}{MN}$. Note that the diagonal elements of $\mathbf{B}$ are all $\alpha_1$, $\alpha_1=\frac{\operatorname{trace}(\mathbf{B})}{MN}$ and we only need to find the trace of $\mathbf{B}$. Since $\mathbf{A}^\intercal(\mathbf{A}\mathbf{A}^\intercal)^{-1}\mathbf{A}=\mathbf{A}^\dagger \mathbf{A}$ is idempotent matrix, $\operatorname{trace}(\mathbf{A}^\dagger \mathbf{A})=\operatorname{rank}(\mathbf{A}^\dagger \mathbf{A})=\operatorname{rank}(\mathbf{A})=(M-1)(N-1)$, we have
\begin{align}
    \operatorname{trace}(\mathbf{B})&=\operatorname{trace}(\mathbf{I}_{MN}-\mathbf{A}^\dagger \mathbf{A})
    \nonumber \\
    &=\operatorname{trace}(\mathbf{I}_{MN})-\operatorname{trace}(\mathbf{A}^\dagger \mathbf{A}) \nonumber\\
    &=MN-(M-1)(N-1)\nonumber \\
    &=M+N-1.
\end{align}
Therefore,
\begin{align}
    \alpha_1=\frac{\operatorname{trace}(\mathbf{B})}{MN}=\frac{M+N-1}{MN}.
\end{align}
Now we consider $\alpha_2$, $\alpha_3$ and $\alpha_4$. It's not hard to see that the number of elements of $\mathbf{B}$ in Type 1, 2, 3, 4 are $1$, $N-1$, $M-1$ and $(M-1)(N-1)$, respectively, and according to Lemma \ref{lemma_1}, we have $\mathbf{B}\mathbf{1}_{MN\times1}=\mathbf{1}_{MN\times1}$, $\mathbf{B}\mathbf{B}^\intercal=\mathbf{B}^2=\mathbf{B}$, $\mathbf{B}\mathbf{A}^\intercal=\mathbf{0}_{MN\times(M-1)(N-1)}$, therefore
\begin{align}
    & \alpha_1+(N-1)\alpha_2+(M-1)\alpha_3 +(M-1)(N-1)\alpha_4=1, \nonumber\\
    & \alpha_1^2+(N-1)\alpha_2^2+(M-1)\alpha_3^2+(M-1)(N-1)\alpha_4^2=\alpha_1, \nonumber \\
     & \alpha_1-\alpha_2-\alpha_3+\alpha_4=0. 
\end{align}
By solving the above equations, we have
\begin{align}
    &\alpha_1=\frac{M+N-1}{MN}, 
     \ \  \ \alpha_2=\frac{M-1}{MN}, \nonumber \\
    &\alpha_3=\frac{N-1}{MN}, 
     \ \ \  \ \ \ \ \ \ \  \alpha_4=-\frac{1}{MN}.
\end{align}

For the monostatic configuration, $M=N$ and we have $\alpha_1=\frac{2M-1}{M^2}$, $\alpha_2=\alpha_3=\frac{M-1}{M^2}$ and $\alpha_4=-\frac{1}{M^2}$. Therefore, Lemma \ref{Lemma_z} holds.
\end{proof}

\section{Performance analysis}
In this section, we study the performance of the proposed estimators and show that by how much the enhancement can be achieved. We also derive the CRLBs for MIMO backscatter TOA estimation and show that the proposed estimators achieve the bounds.  



\subsection{The performance of the proposed estimator}
In this subsection, we derive the performance of the proposed estimators for the bistatic and the monostatic typologies and summarize them in Theorem \ref{thebi} and Theorem \ref{themo}, respectively.
\subsubsection{Bistatic configuration}
\begin{Theorem}\label{thebi}
For the general $M\times N$ bistatic configuration, assuming the noises on the receivers are independent and identically distributed (i.i.d.), then the MSE of the proposed estimator in \eqref{TOAbicloseform} is 
\begin{align}
\frac{M+N-1}{MN}\sigma^2_0,
\end{align}
where $\sigma^2_0$ is the MSE of the conventional least square estimator in \eqref{TOAconvcloseform}. 
\end{Theorem}

\begin{proof}
For the $z$-th element $t_z$ of $\mathbf{t}$, where $z\in\{1,2,\dots,MN\}$, the MSE of conventional LS estimator is
\begin{align}
    \operatorname{E}\left\{(t_z-\hat{t}_z)^2\right\}=\sigma_z^2,
\end{align}
where $\hat{t}_z$ is the $z$-th element of $\hat{\mathbf{t}}$, the MSE of the proposed estimator is:
\begin{align}
    &\operatorname{E}\left\{\left(t_z-\widetilde{t}_{z}\right)^2\right\}
    =\operatorname{E}\left\{\left(t_z-\mathbf{b}_z\hat{\mathbf{t}}\right)^2\right\} \nonumber \\
    & \qquad=\operatorname{E}\left\{\left(t_z-\sum_{r} [\mathbf{B}]_{z,r}\hat{t}_r\right)^2\right\} \nonumber\\
    &\qquad=\operatorname{E}\Bigg\{\Bigg([\mathbf{B}]_{z,z}(t_z-\hat{t}_z)  \nonumber \\
    & \qquad \qquad \quad \left. \left. -\left(([\mathbf{B}]_{z,z}-1){t}_z+\sum_{r\neq z}[\mathbf{B}]_{z,r}\hat{t}_r\right)\right)^2\right\}, 
\end{align}
where $\widetilde{t}_{z}$ is the $z$-th element of $\widetilde{\mathbf{t}}$ and $\mathbf{b}_z=([\mathbf{B}]_{z,1}, [\mathbf{B}]_{z,2}, \dots, [\mathbf{B}]_{z,MN})$ is the $z$-th row of $\mathbf{B}$, $r\in \{1,2,\dots, MN\}$. Let $u=[\mathbf{B}]_{z,z}(t_z-\hat{t}_z)$ and $v=([\mathbf{B}]_{z,z}-1){t}_z+\sum_{r\neq z}[\mathbf{B}]_{z,r}\hat{t}_r$, then
\begin{align}
    \operatorname{E}\left\{\left(t_z-\widetilde{t}_{z}\right)^2\right\}=&\operatorname{E}\left\{(u-v)^2\right\} \nonumber \\
    =&\operatorname{E}\left\{u^2\right\}+\operatorname{E}\left\{2uv\right\}+\operatorname{E}\left\{v^2\right\}.
\end{align}
Note that
\begin{align}
   \operatorname{E}\left\{u^2\right\} &=\operatorname{E}\left\{([\mathbf{B}]_{z,z}(t_z-\hat{t}_z))^2\right\} \nonumber \\
    &=[\mathbf{B}]_{z,z}^2\operatorname{E}\left\{(t_z-\hat{t}_z)^2\right\}
    =[\mathbf{B}]_{z,z}^2\sigma_z^2,
\end{align}
and    
\begin{align}    
    \operatorname{E}\left\{uv\right\}&=\operatorname{E}\left\{u\right\}\operatorname{E}\left\{v\right\} =0
\end{align}
since $\operatorname{E}\left\{u\right\}=\operatorname{E}\Bigg\{[\mathbf{B}]_{z,z}(t_z-\hat{t}_z)\Bigg\}=0$,
and
\begin{align}    
    \operatorname{E}\left\{v^2\right\}&=\operatorname{E}\left\{\left(([\mathbf{B}]_{z,z}-1){t}_z+\sum_{r\neq z}[\mathbf{B}]_{z,r}\hat{t}_r\right)^2\right\} \nonumber \\
    &= \operatorname{E}\Bigg\{\Bigg(([\mathbf{B}]_{z,z}-1){t}_z+\sum_{r\neq z}[\mathbf{B}]_{z,r}{t}_r 
    \nonumber \\
    & \qquad \quad +\sum_{r\neq z}[\mathbf{B}]_{z,r}(\hat{t}_r-t_r)\Bigg)^2\Bigg\} \nonumber \\
    &= \operatorname{E}\left\{\left(\mathbf{b}_z{\mathbf{t}}-t_z+\sum_{r\neq z}[\mathbf{B}]_{z,r}(\hat{t}_r-t_r)\right)^2\right\} \nonumber \\
    &= \operatorname{E}\left\{\left(\sum_{r\neq z}[\mathbf{B}]_{z,r}(\hat{t}_r-t_r)\right)^2\right\}\nonumber \\ &=\sum_{r\neq z}[\mathbf{B}]^2_{z,r}\sigma_{r}^2.
\end{align}
Therefore, the MSE of the proposed estimator is:
\begin{align}
   \operatorname{E}\left\{\left(t_z-\widetilde{t}_{z}\right)^2\right\}&=[\mathbf{B}]^2_{z,z}\sigma_{z}^2+\sum_{r\neq z}[\mathbf{B}]^2_{z,r}\sigma_{r}^2\nonumber \\ &=\sum_r[\mathbf{B}]_{z,r}^2\sigma_{r}^2.
    \label{thederbi}
\end{align}
If the noises on the receivers are i.i.d., i.e. $\sigma_1=\sigma_2=\dots=\sigma_{MN}=\sigma_0$, the MSE of the proposed estimator is $\sum_r[\mathbf{B}]_{z,r}^2\sigma_0^2$. According to Lemma \ref{lemma_1}, we have $\mathbf{B}^\intercal\mathbf{B}=\mathbf{B}^2=\mathbf{B}$, hence 
\begin{align}
    \sum_r[\mathbf{B}]_{z,r}^2=\mathbf{b}_z\mathbf{b}_z^\intercal=[\mathbf{B}]_{z,z}.
\end{align}
According to Lemma \ref{Lemma_z}, we have
\begin{align}
    [\mathbf{B}]_{z,z}=\alpha_1=\frac{M+N-1}{MN},
\end{align}
so the MSE of the proposed estimator is
\begin{align}
    \operatorname{E}\left\{\left(t_z-\widetilde{t}_{z}\right)^2\right\}=[\mathbf{B}]_{z,z}\sigma_0^2=\frac{M+N-1}{MN}\sigma^2_0. 
    \label{msebi}
\end{align}
Therefore Theorem \ref{thebi} holds.
\end{proof}
It is worth mentioning here that for just independent case, equation \eqref{thederbi} is also applicable. For example, when $M=N=2$, for $z=1$, i.e., for the first channel, the MSE is given by
\begin{align}
\operatorname{E}\left\{\left(t_1-\widetilde{t}_{1}\right)^2\right\}&=\alpha_1^2\sigma_1^2+\alpha_3^2\sigma_2^2 +\alpha_2^2\sigma_3^2+\alpha_4^2\sigma_4^2. \nonumber\\&=\frac{9}{16}\sigma_1^2+\frac{1}{16}\sigma_2^2+\frac{1}{16}\sigma_3^2+\frac{1}{16}\sigma_4^2.
\end{align}

\subsubsection{Monoistatic configuration}
\begin{Theorem}\label{themo}
For the general $M\times M$ monostatic configuration, assuming the noises on the receivers are i.i.d., then the MSE of the proposed estimator in \eqref{TOAmocloseform} is 
\begin{align}
    \frac{2M-1}{M^2}\sigma_0^2
\end{align}
for the diagonal subchannels, and
\begin{align}
    \frac{M-1}{M^2}\sigma_0^2
\end{align}
for the off-diagonal subchannels, where $\sigma_0^2$ is the MSE of the conventional least square estimator in \eqref{TOAconvcloseform}.  
\end{Theorem}
\begin{proof}
For the diagonal elements of $\mathbf{T}$ (the TOA of the transceiver received signal from itself), the second constraint in (\ref{TDOAmoConsObjectF}) has no effect on the result, so performance is same to the bistatic configuration when $M=N$ and the MSE is  $\frac{2M-1}{M^2}\sigma_0^2$. For the off-diagonal elements of $\mathbf{T}$ (the TOA of the transceiver received signal from other transceivers) $t_z$, let $t_{\overline{z}}$ be the $\overline{z}$-th parameter in $\mathbf{t}$, and $t_{\overline{z}}$ and $t_z$ are symmetrical elements in $\mathbf{T}$, where $\mathbf{T}=\text{unvec}_{M\times M}(\mathbf{t})$. Then the MSE of $t_z$ is
\begin{align}
    \operatorname{E}\left\{\left(t_z-\widetilde{t}_{z}\right)^2\right\}=&\operatorname{E}\left\{\left(t_z-\frac{\mathbf{b}_z\hat{\mathbf{t}}+\mathbf{b}_{\overline{z}}\hat{\mathbf{t}}}{2}\right)^2\right\}\nonumber\\
    =&\operatorname{E}\left\{\left(t_z-\sum_r\frac{[\mathbf{B}]_{z,r}+[\mathbf{B}]_{\overline{z},r}}{2}\hat{t}_r\right)^2\right\} \nonumber \\
    =& \sum_r\left(\frac{[\mathbf{B}]_{z,r}+[\mathbf{B}]_{\overline{z},r}}{2}\right)^2\sigma_r^2.
    \label{thedermo}
\end{align}
The derivation is similar to (\ref{thederbi}). 

If the noises are i.i.d., i.e., $\sigma_1=\sigma_2=\dots=\sigma_{MN}=\sigma_0$. According to Lemma \ref{lemma_1}, $\mathbf{B}\mathbf{B}^\intercal=\mathbf{B}^2=\mathbf{B}$, hence
\begin{align}
    &\sum_r\left(\frac{[\mathbf{B}]_{z,r}+[\mathbf{B}]_{\overline{z},r}}{2}\right)^2\nonumber \\
    & \qquad =\frac{1}{4}\left(\sum_r [\mathbf{B}]_{z,r}^2+\sum_r[\mathbf{B}]_{{\overline{z}},r}^2+2\sum_r[\mathbf{B}]_{z,r}[\mathbf{B}]_{{\overline{z}},r}\right) \nonumber \\
    &\qquad =\frac{1}{4}([\mathbf{B}]_{z,z}+[\mathbf{B}]_{\overline{z},\overline{z}}+2[\mathbf{B}]_{z,\overline{z}}).
\end{align}
According to Lemma \ref{Lemma_z}, we have
\begin{align}
    &[\mathbf{B}]_{z,\overline{z}}=[\mathbf{B}]_{\overline{z},z}=\alpha_4=-\frac{1}{M^2}, \nonumber \\
    &[\mathbf{B}]_{z,z}=[\mathbf{B}]_{\overline{z},\overline{z}}=\alpha_1=\frac{2M-1}{M^2},
\end{align}
so the MSE of the off-diagonal elements of the proposed estimator is 
\begin{align}
    \operatorname{E}\left\{\left(t_z-\widetilde{t}_{z}\right)^2\right\}=&\left(\frac{1}{4}([\mathbf{B}]_{z,z}+[\mathbf{B}]_{\overline{z},\overline{z}}+2[\mathbf{B}]_{z,\overline{z}})\right)\sigma_0^2 \nonumber \\
    =&\frac{M-1}{M^2}\sigma_0^2.
\end{align}
Therefore Theorem \ref{themo} holds.
\end{proof}

It is worth mentioning here that for the just independent case, equation \eqref{thedermo} is also applicable. For example, when $M=2$, for the off-diagonal channel ($z=2$, $\overline{z}=3$), the MSE is given by
\begin{align}
    \operatorname{E}\left\{\left(t_2-\widetilde{t}_{2}\right)^2\right\}&=(\frac{\alpha_2+\alpha_3}{2})^2\sigma_1^2+(\frac{\alpha_1+\alpha_4}{2})^2\sigma_2^2 \nonumber \\& \ \  \ +(\frac{\alpha_4+\alpha_1}{2})^2\sigma_3^2+(\frac{\alpha_2+\alpha_3}{2})^2\sigma_4^2. \nonumber\\&=\frac{1}{16}\sigma_1^2+\frac{1}{16}\sigma_2^2+\frac{1}{16}\sigma_3^2+\frac{1}{16}\sigma_4^2.
\end{align}

Theorems \ref{thebi} and \ref{themo} showed that the proposed estimator can refine precision from the topological information of the MIMO backscatter channels, and the precision increases with the increase of the MIMO scale, i.e., the increase of $M$ and $N$.


\subsection{Cramer-Rao lower bounds}
In this section, we derive the CRLB for TOA estimation of the MIMO backscatter channels, and show that the proposed estimator is optimal.

\subsubsection{Bistatic configuration}
Since $\mathbf{t}$ is a transformation of $(\mathbf{h}^\intercal, \mathbf{g}^\intercal)^\intercal$ (for the bistatic configuration) or $\mathbf{h}$ (for the monostatic configuration), in general, we can obtain the CRLB of the estimation of $\mathbf{t}$ by deriving the CRLB of $(\mathbf{h}^\intercal, \mathbf{g}^\intercal)^\intercal$ or $\mathbf{h}$. In order to derive the CRLB of $(\mathbf{h}^\intercal, \mathbf{g}^\intercal)^\intercal$ or $\mathbf{h}$, we need to derive the  Fisher   information matrix of $(\mathbf{h}^\intercal, \mathbf{g}^\intercal)^\intercal$ or $\mathbf{h}$ first, and the  covariance matrix is the inverse matrix of the Fisher information matrix. However, for the bistatic configuration, the Fisher information matrix is not full rank. Therefore, we cannot derive the CRLB directly in this case, instead, we need to employ the result from \cite{gorman1990lower}. 

\begin{remark}[Theorem 1 in \cite{gorman1990lower}]
\label{TIT}
Let the parameter space $\mathbf{t}$ be defined by the consistent set of equality constraints: $\mathbf{A}\mathbf{t}=\mathbf{0}$, where the constraints are continuously differentiable. Then for any estimator $\hat{\mathbf{t}}$ having mean $\mathbf{m}_{\mathbf{t}}$, the estimator error covariance matrix $\mathbf{C}_{\hat{\mathbf{t}}}$ satisfies the matrix inequality
\begin{align}
    \mathbf{C}_{\hat{\mathbf{t}}} \geq \frac{\partial\mathbf{m}_{\mathbf{t}}}{\partial \mathbf{t}} \mathbf{Q}_{\mathbf{t}} \bm{J}_{\mathbf{t}}^{-1}\frac{\partial\mathbf{m}_{\mathbf{t}}}{\partial \mathbf{t}}^\intercal,
\end{align}
where idempotent matrix $\mathbf{Q}_{\mathbf{t}}$ is given as follow 
\begin{align}
    \mathbf{Q}_{\mathbf{t}}=\mathbf{I}-\bm{J}_{\mathbf{t}}^{-1}\mathbf{A}^\intercal(\mathbf{A}\bm{J}_{\mathbf{t}}^{-1}\mathbf{A}^\intercal)^\dagger\mathbf{A},
\end{align}
where $\bm{J}_{\mathbf{t}}$ is the Fisher information matrix of model without constraints.
\end{remark}

For the MIMO bistatic TOA channel, we transfer the topology structure into the equivalent linear constraints (Lemma \ref{lemma: equivalence}). Since $\bm{J}_{\mathbf{t}}=\frac{\mathbf{S}^\intercal\mathbf{S}}{\sigma^2}=\frac{L}{\sigma^2}\mathbf{I}_{MN}$ and for unbiased estimator, $\frac{\partial\mathbf{m}_{\mathbf{t}}}{\partial \mathbf{t}}=\mathbf{I}_{MN}$ \cite{gorman1990lower}, the covariance matrix for any estimator of ${\mathbf{t}}$ satisfies
\begin{align}
    \mathbf{C}_{\hat{\mathbf{t}}} &\geq\frac{\partial\mathbf{m}_{\mathbf{t}}}{\partial \mathbf{t}} \mathbf{Q}_{\mathbf{t}} \bm{J}_{\mathbf{t}}^{-1}\frac{\partial\mathbf{m}_{\mathbf{t}}}{\partial \mathbf{t}}^\intercal \nonumber \\
    &=\frac{\sigma^2}{L}(\mathbf{I}_{MN}-\mathbf{A}^\intercal(\mathbf{A}\mathbf{A}^\intercal)^\dagger\mathbf{A}).
\end{align}
Since $\sigma_0^2=\frac{\sigma^2}{L}$, the lower bound of for any TOA estimator is
\begin{align}
    \operatorname{Var}(\hat{t}_i)={[\mathbf{C}_{\hat{\mathbf{t}}}]_{i,i}}\geq\frac{M+N-1}{MNL}\sigma^2=\frac{M+N-1}{MN}\sigma_0^2,
\end{align}
Hence, for the bistatic configuration, the proposed estimator achieves CRLB and is optimal in terms of MSE. 

\subsubsection{Monostatic configuration}
 Different from the derivation of CRLB for the bistatic configuration, for the monostatic configuration, the Fisher information matrix is full rank, so we derive the CRLB for the monostatic configuration in a general way, i.e., obtain the CRLB of the estimation of $\mathbf{t}$ by deriving the CRLB of $\mathbf{h}$.

For the conventional least square estimator, $\hat{\mathbf{t}}\sim \mathcal{N}(\mathbf{t}, \mathbf{C}_{\hat{\mathbf{t}}})$, $\mathbf{C}_{\hat{\mathbf{t}}}=\sigma^2(\mathbf{S}^\intercal\mathbf{S})^{-1}=\frac{\sigma^2}{L}\mathbf{I}_{MN}$. For the desired estimation parameter $\mathbf{h}=({h_1}, {h_2}, \dots, {h_M})$, the Fisher information matrix has element \cite{kay1993fundamentals}
\begin{align}
    [\mathbf{J}_\mathbf{h}]_{i,j}=\frac{\partial \mathbf{t}^\intercal}{\partial {h_i}}\mathbf{C}_{\hat{\mathbf{t}}}^{-1}\frac{\partial \mathbf{t}}{\partial {h_j}}.
\end{align}
Then we have
\begin{align}
    \left\{
    \begin{array}{l l l}
    [\mathbf{J}_\mathbf{h}]_{i,j}=\frac{2L}{\sigma^2}(M+1)
    \quad {\text{if $i=j$}}, \\ 
    {[\mathbf{J}_\mathbf{h}]}_{i,j}=\frac{2L}{\sigma^2}
    \quad \ \ \ \ \ \ \ \ \ \ {\text{if $i \neq j$}}.
  \end{array} \right.
\end{align}
The Fisher information matrix is given by
\begin{align}
    \mathbf{J}_\mathbf{h}=\frac{2L}{\sigma^2}\begin{pmatrix}M+1 & 1 & \dots &1 \\
    1 & M+1 & \dots &1 \\
    \vdots&  & \ddots &\vdots \\
    1& 1 & \dots &M+1 \\\end{pmatrix},
\end{align}
so the covariance matrix for any estimator of ${\mathbf{h}}$ is
\begin{align}
    \mathbf{C}_{\hat{\mathbf{h}}}\geq{\mathbf{J}}_{\mathbf{h}}^{-1}=\frac{\sigma^2}{4L}\begin{pmatrix}\frac{2M-1}{M^2} & -\frac{1}{M^2} & \dots &-\frac{1}{M^2} \\
    -\frac{1}{M^2} & \frac{2M-1}{M^2} & \dots &-\frac{1}{M^2} \\
    \vdots&  & \ddots &\vdots \\
    -\frac{1}{M^2}& -\frac{1}{M^2} & \dots &\frac{2M-1}{M^2} \\\end{pmatrix}.
\end{align}
Since $\sigma_0^2=\frac{\sigma^2}{L}$, the CRLB of the diagonal subchannels of $\mathbf{T}$ is
\begin{align}
    \operatorname{Var}(\hat{h}_i+\hat{h}_i)=4\operatorname{Var}(\hat{h}_i)=4{[\mathbf{C}_{\hat{\mathbf{h}}}]}_{i,i}\geq\frac{2M-1}{M^2}\sigma_0^2.
\end{align}
The CRLB of the off-diagonal subchannels of $\mathbf{T}$ is
\begin{align}
    \operatorname{Var}(\hat{h}_i+{\hat{h}_j})&=\operatorname{Var}(\hat{h}_i)+\operatorname{Var}({\hat{h}_j})+2\operatorname{Cov}(\hat{h}_i,{\hat{h}_j})\nonumber \\
    &={[\mathbf{C}_{\hat{\mathbf{h}}}]}_{i,i}+{[\mathbf{C}_{\hat{\mathbf{h}}}]}_{j,j}+2{[\mathbf{C}_{\hat{\mathbf{h}}}]}_{i,j}\nonumber \\
    &\geq\left(\frac{2M-1}{M^2}-\frac{1}{M^2}\right)\frac{\sigma^2_0}{2}\nonumber \\ &=\frac{M-1}{M^2}\sigma_0^2.
\end{align}
Clearly, for the monostatic configuration, the proposed estimator achieves CRLB and is optimal in terms of MSE. 

It is worth mentioning here that both the work in \cite{yansouni1984use} and our work employ geometry/topology constraint to improve the estimation. However, \cite{yansouni1984use} studied the geometry constraint that imposed to the conventional TDOA model, while we studied the topological constraint of the MIMO backscatter TOA model, which, to our best knowledge has never been studied before. As we can see, the geometry constraint found in \cite{yansouni1984use} and the topological constraint found in our work yield completely different constraint matrices, therefore the estimation problems, as well as the results, are all completely different.



\section{Simulations}
In this section, we validate the performance of the proposed TOA estimators, and compare them with the conventional estimator. 
The locations of tag and transceiver antennas are uniformly initialize  in the 3D cubic area with side length $10$ meters, and we set the true TOA is ${t}=\frac{l}{c}$, where $l$ is the distance and $c$ is the speed of light. All the simulation results plotted here have been obtained numerically after averaging over at least $10^4$ independent channel realizations. 

For the bistatic configuration, Fig. \ref{fig: es_bi_M4N3} and Fig. \ref{fig: es_bi_M6N6} show the MSEs of the conventional estimator and the proposed estimator versus $\sigma$ with various antenna settings and training lengths. It is clear that the proposed estimator outperforms  the conventional estimator considerably and the simulation results match well with the theoretical results. The simulation results for the monostatic configuration are given in Fig. \ref{fig: es_mo_M6} and Fig. \ref{fig: es_mo_M8}. It is clear that we have the similar conclusion as that of the bistatic configuration.
\begin{figure}
\centering
{\includegraphics[width=1\columnwidth]{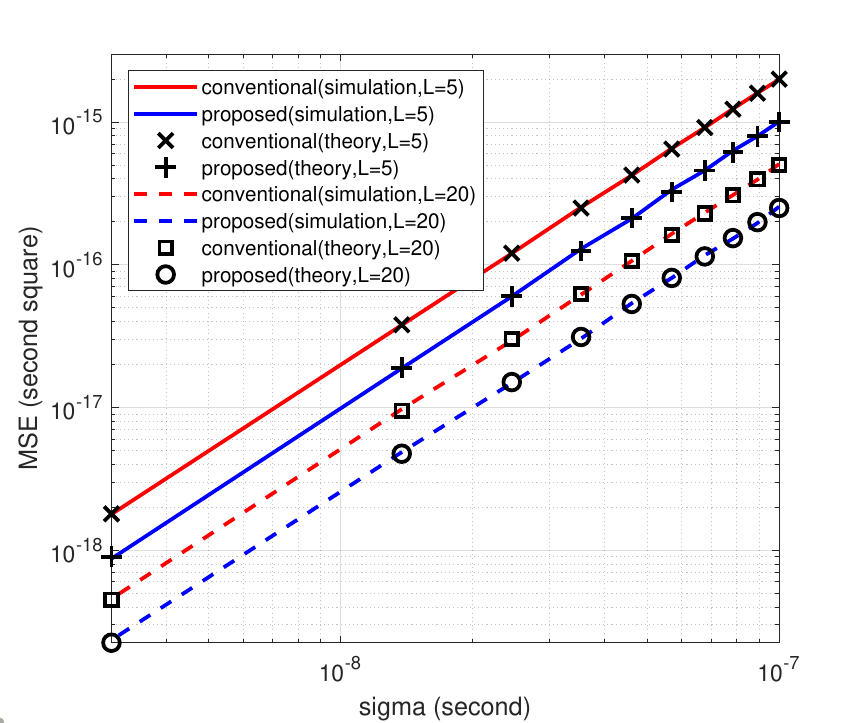}}
\caption{For the bistatic configuration: the MSEs of the conventional estimator and the proposed estimator versus $\sigma$ and training length $L$, with $M=4$, $N=3$. }
\label{fig: es_bi_M4N3}
\end{figure}

\begin{figure}
\centering
{\includegraphics[width=1\columnwidth]{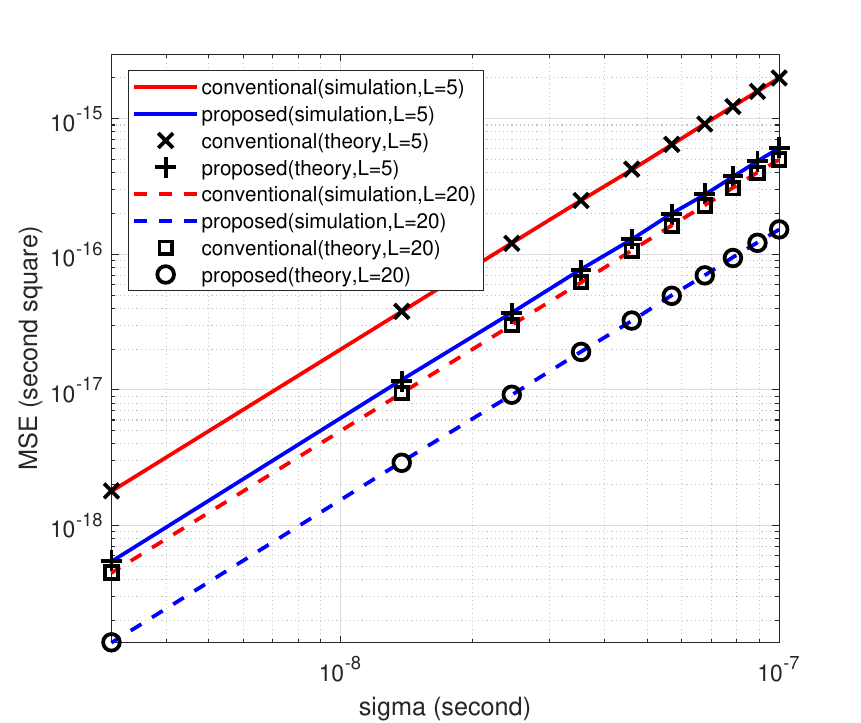}}
\caption{For the bistatic configuration: the MSEs of the conventional estimator and the proposed estimator versus $\sigma$ and training length $L$, with $M=N=6$. }
\label{fig: es_bi_M6N6}
\end{figure}

\begin{figure}
\centering
{\includegraphics[width=1\columnwidth]{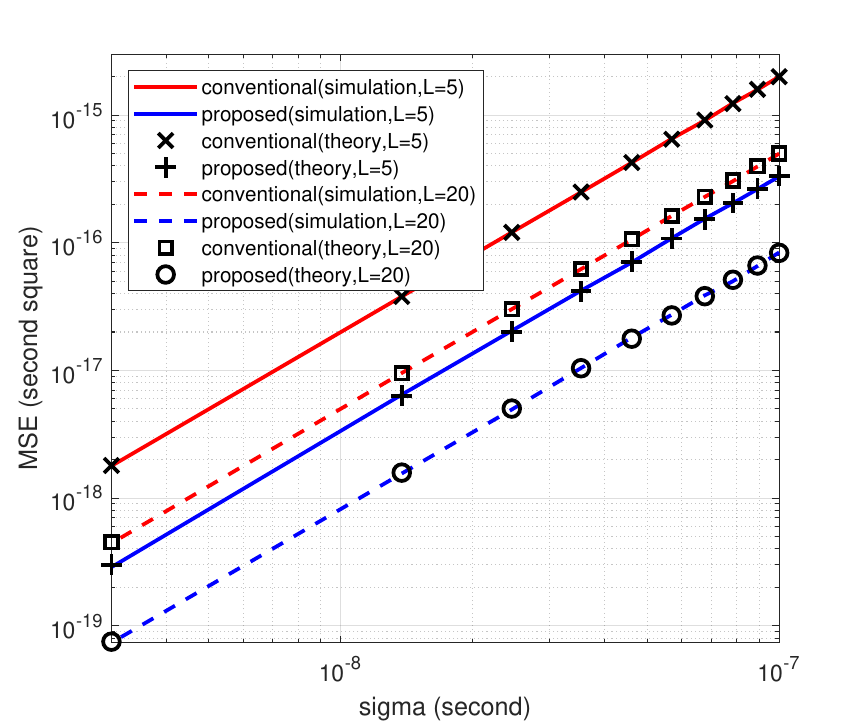}}
\caption{For the monostatic configuration: the MSEs of the conventional estimator and the proposed estimator versus $\sigma$ and training length $L$, with $M=6$. }
\label{fig: es_mo_M6}
\end{figure}

\begin{figure}
\centering
{\includegraphics[width=1\columnwidth]{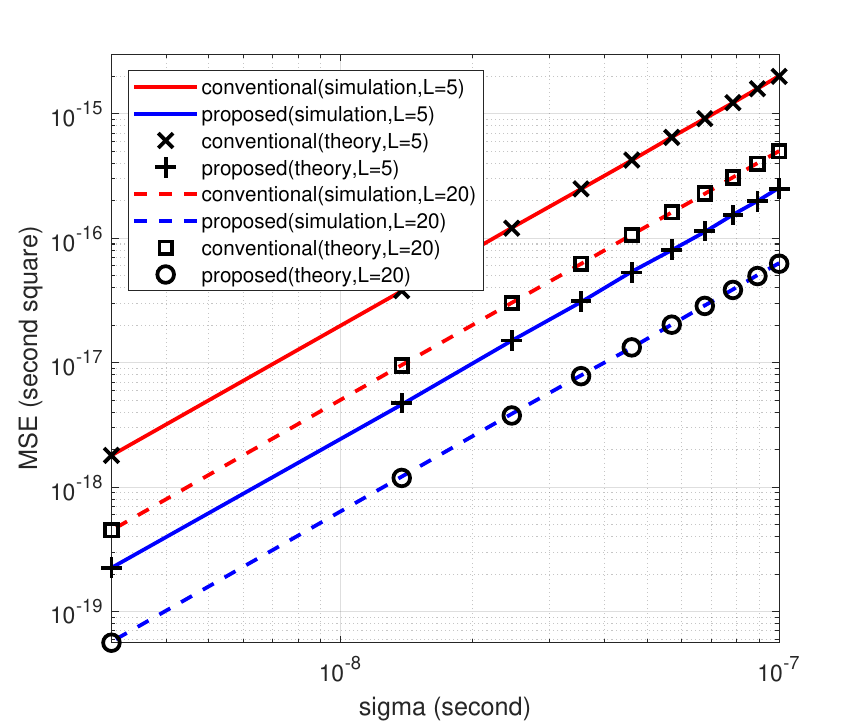}}
\caption{For the monostatic configuration: the MSEs of the conventional estimator and the proposed estimator versus $\sigma$ and training length $L$, with $M=8$. }
\label{fig: es_mo_M8}
\end{figure}

\begin{figure}
\centering
{\includegraphics[width=1\columnwidth]{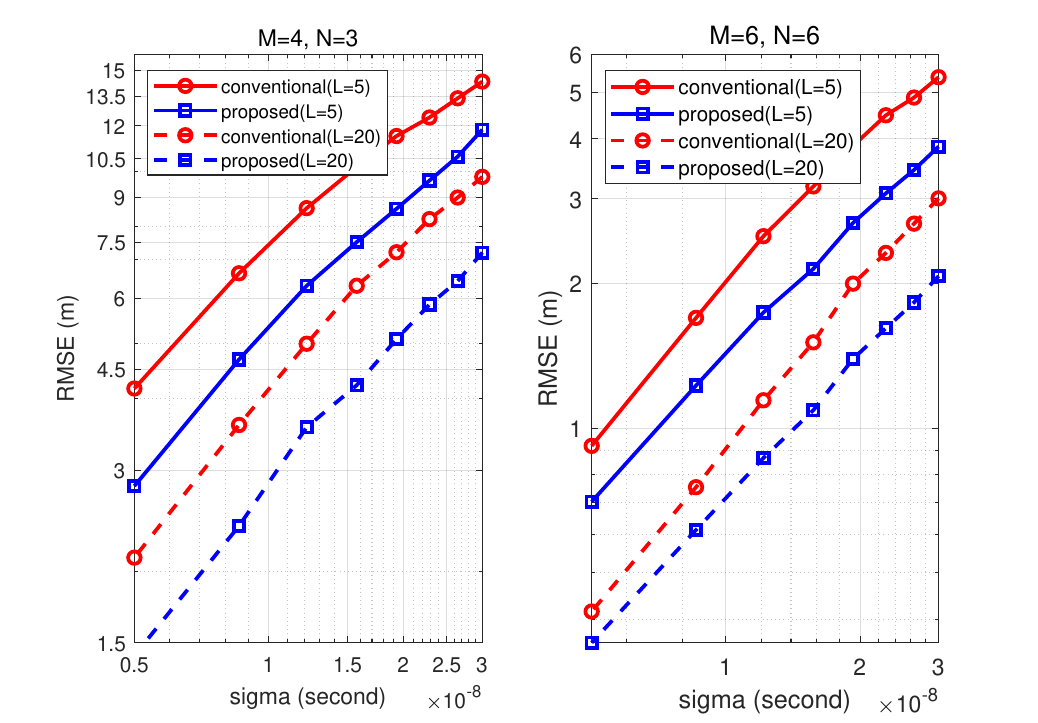}}
\caption{For the bistatic configuration: the RMSEs of TOA localization by employing the conventional estimator and the proposed estimator versus $\sigma$ with various antenna settings and training length.}
\label{fig: lo_bi_M4N3}
\end{figure}

\begin{figure}
\centering
{\includegraphics[width=1\columnwidth]{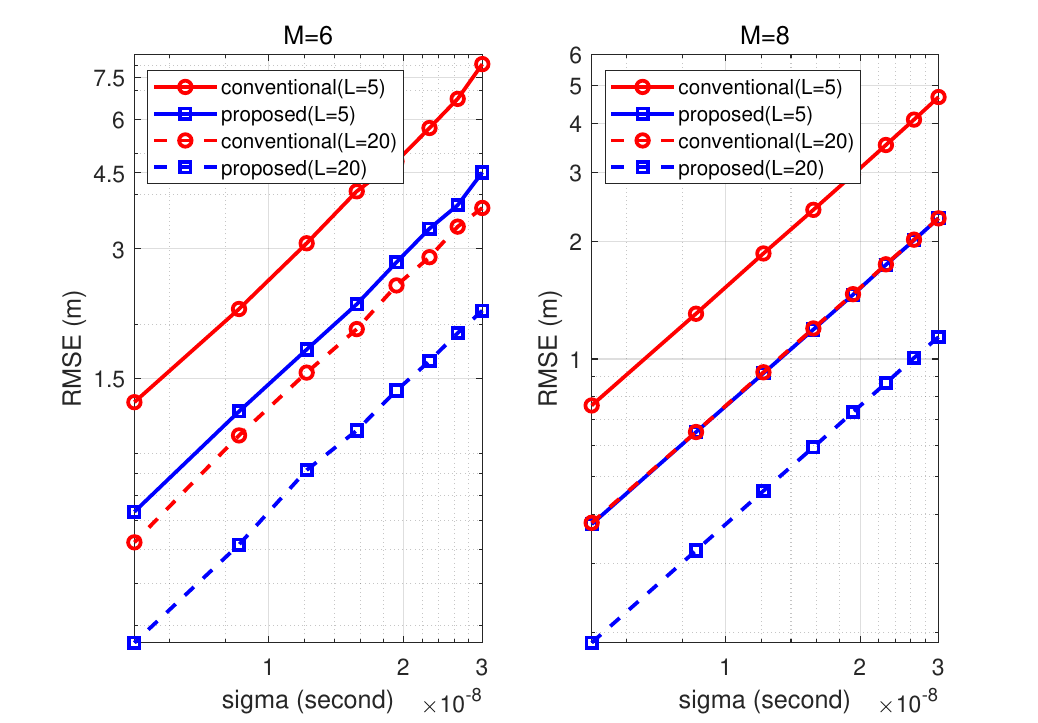}}
\caption{For the monostatic configuration: the RMSEs of TOA localization by employing the conventional estimator and the proposed estimator versus $\sigma$ with various antenna settings and training length.}
\label{fig: lo_mo_M6}
\end{figure}

For localization, we consider the most commonly employed localization algorithm \cite{chan1994simple, einemo2015weighted} for TOA localization. For the bistatic configuration, Fig. \ref{fig: lo_bi_M4N3} shows the root mean square (RMSE) of TOA localization with the conventional TOA estimator and the proposed TOA estimator versus $\sigma$ with various antenna settings and training lengths. We can see that the localization accuracy of the proposed TOA estimator considerably outperforms the conventional TOA estimator. For the monostatic configuration, Fig. \ref{fig: lo_mo_M6} shows the similar conclusion.


\section{Conclusion}
In this paper, we focused on the TOA estimation for the MIMO backscatter channels. By investigating the topological structures of the MIMO backscatter channels, we proposed a novel TOA estimator in closed form that can significantly enhance the estimation accuracy. We showed that the MSE of the proposed estimator for the $M\times N$ bistatic configuration is  $\frac{M+N-1}{MN}\sigma_0^2$, and for the monostatic configuration, the MSE is $\frac{2M-1}{M^2}\sigma_0^2$ for the diagonal subchannels, and $\frac{M-1}{M^2}\sigma_0^2$ for the off-diagonal subchannels, where $\sigma_0^2$ is the MSE of conventional estimator. In addition, we derived the Cramer-Rao lower  bound for the MIMO backscatter TOA estimation to verify that the proposed estimator is optimal. Finally, numerical results confirmed that the proposed estimator outperformed the conventional estimator and can significantly improve the positioning accuracy.

\bibliographystyle{IEEEtran}
\bibliography{IEEE}

\begin{thebibliography}{10}
\providecommand{\url}[1]{#1}
\csname url@samestyle\endcsname
\providecommand{\newblock}{\relax}
\providecommand{\bibinfo}[2]{#2}
\providecommand{\BIBentrySTDinterwordspacing}{\spaceskip=0pt\relax}
\providecommand{\BIBentryALTinterwordstretchfactor}{4}
\providecommand{\BIBentryALTinterwordspacing}{\spaceskip=\fontdimen2\font plus
\BIBentryALTinterwordstretchfactor\fontdimen3\font minus \fontdimen4\font\relax}
\providecommand{\BIBforeignlanguage}[2]{{%
\expandafter\ifx\csname l@#1\endcsname\relax
\typeout{** WARNING: IEEEtran.bst: No hyphenation pattern has been}%
\typeout{** loaded for the language `#1'. Using the pattern for}%
\typeout{** the default language instead.}%
\else
\language=\csname l@#1\endcsname
\fi
#2}}
\providecommand{\BIBdecl}{\relax}
\BIBdecl

\bibitem{he2020monostatic}
C.~He, S.~Chen, H.~Luan, X.~Chen, and Z.~J. Wang, ``{Monostatic MIMO backscatter communications},'' \emph{IEEE Journal on Selected Areas in Communications}, vol.~38, no.~8, pp. 1896--1909, 2020.

\bibitem{luan2021better}
H.~Luan, X.~Xie, L.~Han, C.~He, and Z.~J. Wang, ``{A Better than Alamouti OSTBC for MIMO Backscatter Communications},'' \emph{IEEE Transactions on Wireless Communications}, 2021.

\bibitem{mishra2019optimal}
D.~Mishra and E.~G. Larsson, ``{Optimal channel estimation for reciprocity-based backscattering with a full-duplex MIMO reader},'' \emph{IEEE Transactions on Signal Processing}, vol.~67, no.~6, pp. 1662--1677, 2019.

\bibitem{zhang2018constellation}
Q.~Zhang, H.~Guo, Y.-C. Liang, and X.~Yuan, ``{Constellation learning-based signal detection for ambient backscatter communication systems},'' \emph{IEEE Journal on Selected Areas in Communications}, vol.~37, no.~2, pp. 452--463, 2018.

\bibitem{sun2020eigenspace}
Y.~Sun, K.~Ho, and Q.~Wan, ``{Eigenspace solution for AOA localization in modified polar representation},'' \emph{IEEE Transactions on Signal Processing}, vol.~68, pp. 2256--2271, 2020.

\bibitem{liu2015rss}
C.~Liu, D.~Fang, Z.~Yang, H.~Jiang, X.~Chen, W.~Wang, T.~Xing, and L.~Cai, ``{RSS distribution-based passive localization and its application in sensor networks},'' \emph{IEEE Transactions on Wireless Communications}, vol.~15, no.~4, pp. 2883--2895, 2015.

\bibitem{park2015closed}
C.-H. Park and J.-H. Chang, ``{Closed-form localization for distributed MIMO radar systems using time delay measurements},'' \emph{IEEE Transactions on Wireless Communications}, vol.~15, no.~2, pp. 1480--1490, 2015.

\bibitem{cheung2004least}
K.~W. Cheung, H.-C. So, W.-K. Ma, and Y.-T. Chan, ``{Least squares algorithms for time-of-arrival-based mobile location},'' \emph{IEEE Transactions on Signal Processing}, vol.~52, no.~4, pp. 1121--1130, 2004.

\bibitem{gillette2008linear}
M.~D. Gillette and H.~F. Silverman, ``{A linear closed-form algorithm for source localization from time-differences of arrival},'' \emph{IEEE Signal Processing Letters}, vol.~15, pp. 1--4, 2008.

\bibitem{sahinoglu2008ultra}
Z.~Sahinoglu, S.~Gezici, and I.~Guvenc, ``{Ultra-wideband positioning systems},'' \emph{Cambridge, New York}, 2008.

\bibitem{win2002characterization}
M.~Z. Win and R.~A. Scholtz, ``{Characterization of ultra-wide bandwidth wireless indoor channels: A communication-theoretic view},'' \emph{IEEE Journal on Selected Areas in Communications}, vol.~20, no.~9, pp. 1613--1627, 2002.

\bibitem{d2008energy}
A.~A. D'amico, U.~Mengali, and L.~Taponecco, ``{Energy-based TOA estimation},'' \emph{IEEE Transactions on Wireless Communications}, vol.~7, no.~3, pp. 838--847, 2008.

\bibitem{bialer2011efficient}
O.~Bialer, D.~Raphaeli, and A.~J. Weiss, ``{Efficient time of arrival estimation algorithm achieving maximum likelihood performance in dense multipath},'' \emph{IEEE Transactions on Signal Processing}, vol.~60, no.~3, pp. 1241--1252, 2011.

\bibitem{giorgetti2013time}
A.~Giorgetti and M.~Chiani, ``{Time-of-arrival estimation based on information theoretic criteria},'' \emph{IEEE Transactions on Signal Processing}, vol.~61, no.~8, pp. 1869--1879, 2013.

\bibitem{kassas2021joint}
Z.~M. Kassas and K.~Shamaei, ``{A Joint TOA and DOA Acquisition and Tracking Approach for Positioning with LTE Signals},'' \emph{IEEE Transactions on Signal Processing}, 2021.

\bibitem{shang2014ml}
F.~Shang, B.~Champagne, and I.~N. Psaromiligkos, ``{A ML-based framework for joint TOA/AOA estimation of UWB pulses in dense multipath environments},'' \emph{IEEE Transactions on Wireless Communications}, vol.~13, no.~10, pp. 5305--5318, 2014.

\bibitem{shamaei2021receiver}
K.~Shamaei and Z.~M. Kassas, ``{Receiver design and time of arrival estimation for opportunistic localization with 5G signals},'' \emph{IEEE Transactions on Wireless Communications}, 2021.

\bibitem{gifford2020impact}
W.~Gifford, D.~Dardari, and M.~Win, ``{The impact of multipath information on time-of-arrival estimation},'' \emph{IEEE Transactions on Signal Processing}, 2020.

\bibitem{lee2002ranging}
J.-Y. Lee and R.~A. Scholtz, ``{Ranging in a dense multipath environment using an UWB radio link},'' \emph{IEEE Journal on Selected Areas in Communications}, vol.~20, no.~9, pp. 1677--1683, 2002.

\bibitem{aditya2018survey}
S.~Aditya, A.~F. Molisch, and H.~M. Behairy, ``{A survey on the impact of multipath on wideband time-of-arrival based localization},'' \emph{Proceedings of the IEEE}, vol. 106, no.~7, pp. 1183--1203, 2018.

\bibitem{falsi2006time}
C.~Falsi, D.~Dardari, L.~Mucchi, and M.~Z. Win, ``{Time of arrival estimation for UWB localizers in realistic environments},'' \emph{EURASIP Journal on Advances in Signal Processing}, vol. 2006, pp. 1--13, 2006.

\bibitem{xu2010toa}
C.~Xu and C.~L. Law, ``{TOA estimator for UWB backscattering RFID system with clutter suppression capability},'' \emph{EURASIP Journal on Wireless Communications and Networking}, vol. 2010, pp. 1--14, 2010.

\bibitem{roemer2010tensor}
F.~Roemer and M.~Haardt, ``{Tensor-based channel estimation and iterative refinements for two-way relaying with multiple antennas and spatial reuse},'' \emph{IEEE Transactions on Signal Processing}, vol.~58, no.~11, pp. 5720--5735, 2010.

\bibitem{sidiropoulos2017tensor}
N.~D. Sidiropoulos, L.~De~Lathauwer, X.~Fu, K.~Huang, E.~E. Papalexakis, and C.~Faloutsos, ``{Tensor decomposition for signal processing and machine learning},'' \emph{IEEE Transactions on Signal Processing}, vol.~65, no.~13, pp. 3551--3582, 2017.

\bibitem{nion2010tensor}
D.~Nion and N.~D. Sidiropoulos, ``{Tensor algebra and multidimensional harmonic retrieval in signal processing for MIMO radar},'' \emph{IEEE Transactions on Signal Processing}, vol.~58, no.~11, pp. 5693--5705, 2010.

\bibitem{lioliou2008least}
P.~Lioliou and M.~Viberg, ``{Least-squares based channel estimation for MIMO relays},'' in \emph{2008 International ITG Workshop on Smart Antennas}.\hskip 1em plus 0.5em minus 0.4em\relax IEEE, 2008, pp. 90--95.

\bibitem{zhou2016channel}
Z.~Zhou, J.~Fang, L.~Yang, H.~Li, Z.~Chen, and S.~Li, ``{Channel estimation for millimeter-wave multiuser MIMO systems via PARAFAC decomposition},'' \emph{IEEE Transactions on Wireless Communications}, vol.~15, no.~11, pp. 7501--7516, 2016.

\bibitem{rong2012channel}
Y.~Rong, M.~R. Khandaker, and Y.~Xiang, ``{Channel estimation of dual-hop MIMO relay system via parallel factor analysis},'' \emph{IEEE Transactions on Wireless Communications}, vol.~11, no.~6, pp. 2224--2233, 2012.

\bibitem{zhao2019channel}
W.~Zhao, G.~Wang, S.~Atapattu, R.~He, and Y.-C. Liang, ``{Channel estimation for ambient backscatter communication systems with massive-antenna reader},'' \emph{IEEE Transactions on Vehicular Technology}, vol.~68, no.~8, pp. 8254--8258, 2019.

\bibitem{gholami2012improved}
M.~R. Gholami, S.~Gezici, and E.~G. Strom, ``{Improved position estimation using hybrid TW-TOA and TDOA in cooperative networks},'' \emph{IEEE Transactions on Signal Processing}, vol.~60, no.~7, pp. 3770--3785, 2012.

\bibitem{nguyen2016optimal}
N.~H. Nguyen and K.~Do{\u{g}}an{\c{c}}ay, ``{Optimal geometry analysis for multistatic TOA localization},'' \emph{IEEE Transactions on Signal Processing}, vol.~64, no.~16, pp. 4180--4193, 2016.

\bibitem{coluccia2017hybrid}
A.~Coluccia and A.~Fascista, ``{On the hybrid TOA/RSS range estimation in wireless sensor networks},'' \emph{IEEE Transactions on Wireless Communications}, vol.~17, no.~1, pp. 361--371, 2017.

\bibitem{hassibi2003much}
B.~Hassibi and B.~M. Hochwald, ``{How much training is needed in multiple-antenna wireless links?}'' \emph{IEEE Transactions on Information Theory}, vol.~49, no.~4, pp. 951--963, 2003.

\bibitem{kay1993fundamentals}
S.~M. Kay, \emph{{Fundamentals of statistical signal processing: estimation theory}}.\hskip 1em plus 0.5em minus 0.4em\relax Prentice-Hall, Inc., 1993.

\bibitem{gorman1990lower}
J.~D. Gorman and A.~O. Hero, ``{Lower bounds for parametric estimation with constraints},'' \emph{IEEE Transactions on Information Theory}, vol.~36, no.~6, pp. 1285--1301, 1990.

\bibitem{yansouni1984use}
P.~Yansouni and R.~Inkol, ``{The use of linear constraints to reduce the variance of time of arrival difference estimates for source location},'' \emph{IEEE Transactions on Acoustics, Speech, and Signal Processing}, vol.~32, no.~4, pp. 907--912, 1984.

\bibitem{chan1994simple}
Y.-T. Chan and K.~Ho, ``A simple and efficient estimator for hyperbolic location,'' \emph{IEEE Transactions on signal processing}, vol.~42, no.~8, pp. 1905--1915, 1994.

\bibitem{einemo2015weighted}
M.~Einemo and H.~C. So, ``{Weighted least squares algorithm for target localization in distributed MIMO radar},'' \emph{Signal Processing}, vol. 115, pp. 144--150, 2015.

\end{thebibliography}

\appendix

\begin{proof}

In order to prove solution of \eqref{TOAmocloseform} is also in the feasible set defined by the first constraint in \eqref{TDOAmoConsObjectF}, we only need to prove each $2\times 2$ submatrix $\bar{\mathbf{D}}$ of $\bar{\mathbf{T}}$ satisfied $\begin{pmatrix}
    1 & -1 & -1 & 1
    \end{pmatrix}\text{vec}(\bar{\mathbf{D}}) = 0$.
Under the first constraint solely, let $2\times 2$ submatrix of $\bar{\mathbf{T}}$ be
\begin{align}
    \bar{\mathbf{D}}=\begin{pmatrix}
    {[\bar{\mathbf{T}}]}_{p_x,p_y} &  {[\bar{\mathbf{T}}]}_{p_x,p_y+1} \\ {[\bar{\mathbf{T}}]}_{p_x+1,p_y} &  {[\bar{\mathbf{T}}]}_{p_x+1,p_y+1}
    \end{pmatrix},
\end{align}
where $p_x\in \{1, \dots, M - 1\}$ and $p_y\in \{1, \dots, N - 1\}$. 
Then the solution by the second constraint can be expressed as
\begin{align}
    \widetilde{\mathbf{D}} =\begin{pmatrix}
    \frac{\bar{[\mathbf{T}]}_{p_x,p_y}+\bar{[\mathbf{T}]}_{p_y,p_x}}{2} &  \frac{\bar{[\mathbf{T}]}_{p_x,p_y+1}+\bar{[\mathbf{T}]}_{p_y+1,p_x}}{2} \\ \frac{\bar{[\mathbf{T}]}_{p_x+1,p_y}+\bar{[\mathbf{T}]}_{p_y,p_x+1}}{2} &  \frac{\bar{[\mathbf{T}]}_{p_x+1,p_y+1}+\bar{[\mathbf{T}]}_{p_y+1,p_x+1}}{2}
    \end{pmatrix}.
\end{align}

Since the solution by the first constraint satisfied
\begin{align}
    \bar{[\mathbf{T}]}_{p_x,p_y} + \bar{[\mathbf{T}]}_{p_x+1,p_y+1} = \bar{[\mathbf{T}]}_{p_x,p_y+1} + \bar{[\mathbf{T}]}_{p_x+1,p_y},
\end{align}
\begin{align}
    \bar{[\mathbf{T}]}_{p_y,p_x} + \bar{[\mathbf{T}]}_{p_y+1,p_x+1} = \bar{[\mathbf{T}]}_{p_y,p_x+1} + \bar{[\mathbf{T}]}_{p_y+1,p_x},
\end{align}
then we have
\begin{align}
     &\frac{\bar{[\mathbf{T}]}_{p_x,p_y}+\bar{[\mathbf{T}]}_{p_y,p_x}}{2}+\frac{\bar{[\mathbf{T}]}_{p_x+1,p_y+1}+\bar{[\mathbf{T}]}_{p_y+1,p_x+1}}{2} \nonumber\\
    =&\frac{\bar{[\mathbf{T}]}_{p_x,p_y+1}+\bar{[\mathbf{T}]}_{p_y+1,p_x}}{2}+\frac{\bar{[\mathbf{T}]}_{p_x+1,p_y}+\bar{[\mathbf{T}]}_{p_y,p_x+1}}{2},
\end{align}
i.e., each submatrix $\bar{\mathbf{D}}$ satisfied
\begin{align}
    \begin{pmatrix}
    1 & -1 & -1 & 1
    \end{pmatrix}\text{vec}(\bar{\mathbf{D}}) = 0.
\end{align}
Therefore, the solution by the second constraint still satisfies the first constraint.

\end{proof}

\end{document}